\documentclass[a4paper,onecolumn,superscriptaddress,11pt]{quantumarticle}
\usepackage{graphicx}
\usepackage[pdftex,colorlinks=true,linkcolor=blue,citecolor=blue,urlcolor=blue]{hyperref}
\usepackage{amsmath, amsthm, amssymb}
\usepackage{subfigure}
\usepackage{comment}
\usepackage{url}
\usepackage[ruled,lined,linesnumbered]{algorithm2e}
\usepackage{tikz}
\usetikzlibrary{shapes}

\newcommand{\R}{\mathbb{R}}

\newcommand{\C}{\mathbb{C}}
\newcommand{\Z}{\mathbb{Z}}

\newcommand{\E}{\mathbb{E}}
\newcommand{\F}{\mathbb{F}}

\newcommand{\proj}[1]{| #1 \rangle \langle #1 |}
\newcommand{\bracket}[3]{\langle #1|#2|#3 \rangle}
\newcommand{\braket}[1]{\langle #1 \rangle}
\newcommand{\sm}[1]{\left( \begin{smallmatrix} #1 \end{smallmatrix} \right)}

\DeclareMathOperator{\poly}{poly}

\DeclareMathOperator{\diag}{diag}

\DeclareMathOperator{\NP}{NP}
\DeclareMathOperator{\FBPP}{FBPP}

\newcommand{\be}{\begin{equation}}
\newcommand{\ee}{\end{equation}}
\newcommand{\bea}{\begin{eqnarray}}
\newcommand{\eea}{\end{eqnarray}}
\newcommand{\bes}{\begin{equation*}}
\newcommand{\ees}{\end{equation*}}
\newcommand{\beas}{\begin{eqnarray*}}
\newcommand{\eeas}{\end{eqnarray*}}

\makeatletter
\newtheorem*{rep@theorem}{\rep@title}
\newcommand{\newreptheorem}[2]{%
\newenvironment{rep#1}[1]{%
 \def\rep@title{#2 \ref{##1} (restated)}%
 \begin{rep@theorem}}%
 {\end{rep@theorem}}}
\makeatother

\newtheorem{thm}{Theorem}
\newtheorem*{thm*}{Theorem}
\newtheorem{cor}[thm]{Corollary}
\newtheorem{con}[thm]{Conjecture}
\newtheorem{lem}[thm]{Lemma}
\newtheorem*{lem*}{Lemma}

\newreptheorem{thm}{Theorem}
\newreptheorem{lem}{Lemma}

\input{Qcircuit.tex}

\begin{document}

\title{Achieving quantum supremacy with sparse and noisy commuting quantum computations}
\author{Michael J. Bremner}
\affiliation{Centre for Quantum Computation and Communication Technology, Centre for Quantum Software and Information, Faculty of Engineering and Information Technology, University of Technology Sydney, NSW 2007, Australia.}
\author{Ashley Montanaro}
\affiliation{School of Mathematics, University of Bristol, UK}
\email{ashley.montanaro@bristol.ac.uk}
\author{Dan J. Shepherd}
\affiliation{NCSC, Hubble Road, Cheltenham, UK.}
\maketitle

\begin{abstract}
The class of commuting quantum circuits known as IQP (instantaneous quantum polynomial-time) has been shown to be hard to simulate classically, assuming certain complexity-theoretic conjectures. Here we study the power of IQP circuits in the presence of physically motivated constraints. First, we show that there is a family of sparse IQP circuits that can be implemented on a square lattice of $n$ qubits in depth $O(\sqrt{n} \log n)$, and which is likely hard to simulate classically. Next, we show that, if an arbitrarily small constant amount of noise is applied to each qubit at the end of any IQP circuit whose output probability distribution is sufficiently anticoncentrated, there is a polynomial-time classical algorithm that simulates sampling from the resulting distribution, up to constant accuracy in total variation distance. However, we show that purely classical error-correction techniques can be used to design IQP circuits which remain hard to simulate classically, even in the presence of arbitrary amounts of noise of this form. These results demonstrate the challenges faced by experiments designed to demonstrate quantum supremacy over classical computation, and how these challenges can be overcome.
\end{abstract}


\section{Introduction}

Over the last few years there has been significant attention devoted to devising experimental demonstrations of {\em quantum supremacy}~\cite{preskill12}: namely a quantum computer solving a computational task that goes beyond what a classical machine could achieve. This is, in part, driven by the hope that a clear demonstration of quantum supremacy can be performed with a device that is intermediate between the small quantum circuits that can currently be built and a full-scale quantum computer. The theoretical challenge that this poses is twofold: firstly we must identify the physically least expensive quantum computations that are classically unachievable; and we must also determine if this advantage can be maintained in the presence of physical noise. 

There are several intermediate quantum computing models which could be used to demonstrate quantum supremacy, including simple linear-optical circuits (the boson sampling problem~\cite{aaronson13}); the one clean qubit model~\cite{morimae14}; and commuting quantum circuits, a model known as ``IQP''~\cite{shepherd09,bremner11}. In each of these cases, it has been shown that efficient classical simulation of the simple quantum computations involved is not possible, assuming that the polynomial hierarchy does not collapse. However, these results only prove hardness of simulating the {\em ideal} quantum computations in question up to a small relative error in each output probability.

Any quantum experiment will be subject to noise, and the noisy experiment could be substantially easier to simulate than the noise-free experiment. In an attempt to address this, it was shown in~\cite{aaronson13,bremner16} that, assuming certain additional complexity-theoretic conjectures, the probability distributions resulting from boson sampling and IQP circuits are still hard to sample from classically, even up to small total variation distance. For example, in~\cite{bremner16} the following two conjectures were made, one native to condensed-matter physics, the other to computer science:

\begin{con}
\label{con:isingorig}
Consider the partition function of the general Ising model,
\be \label{eq:isingorig} Z(\omega) = \sum_{z \in \{\pm1\}^n} \omega^{\sum_{i<j} w_{ij} z_i z_j + \sum_{k=1}^n v_k z_k }, \ee
where the exponentiated sum is over the complete graph on $n$ vertices, $w_{ij}$ and $v_k$ are real edge and vertex weights, and $\omega \in \C$. Let the edge and vertex weights be picked uniformly at random from the set $\{0,\dots,7\}$.

Then it is \#P-hard to approximate $|Z(e^{i\pi/8})|^2$ up to multiplicative error $1/4 + o(1)$ for a $1/24$ fraction of instances, over the random choice of weights.
\end{con}

\begin{con}
\label{con:poly}
Let $f:\{0,1\}^n \rightarrow \{0,1\}$ be a uniformly random degree-3 polynomial over $\F_2$, and define $\operatorname{ngap}(f) := (|\{x: f(x)=0\}| - |\{x: f(x)=1\}|)/2^n$. Then it is \#P-hard to approximate $\operatorname{ngap}(f)^2$ up to a multiplicative error of $1/4 + o(1)$ for a $1/24$ fraction of polynomials $f$.
\end{con}

It was shown in~\cite{bremner16} that, if we assume either Conjecture \ref{con:isingorig} or Conjecture \ref{con:poly}, and the widely-believed complexity-theoretic assumption that the polynomial hierarchy does not collapse, then there is no polynomial-time classical algorithm for approximately sampling from the output distributions of IQP circuits. That is, if $p$ is the distribution that the noise-free quantum circuit would produce, it is hard for the classical machine to sample from any distribution $p'$ such that $\| p - p' \|_1 \le \epsilon$, for some small $\epsilon$, where the size of $\epsilon$ depends on the conjectures one is willing to assume. These results imply that a fault-tolerant implementation of IQP sampling or boson sampling can be made resilient to noise while (potentially) maintaining a quantum advantage.

Although this was a significant step towards the near-term possibility of quantum supremacy, these results still suffer from some shortcomings:
\begin{enumerate}
\item They do not yet resolve the question of whether realistically noisy, and non-fault-tolerant, quantum experiments are likely to be hard to simulate classically. Indeed, applying a small amount of independent noise to each qubit can readily lead to a distribution $p'$ which is much further from $p$ than the regime in which the results of~\cite{aaronson13,bremner16} are applicable.
\item The results of~\cite{bremner16} assume that all pairs of qubits are able to interact. Such long-range interactions incur significant physical resource overheads for most computational architectures.
\end{enumerate}


\subsection{Our results}

Here we study the behaviour of IQP circuits which are implemented on hardware with spatial locality constraints, and in the presence of noise. These are critical questions for any realistic experimental implementation.

An IQP circuit (``Instantaneous Quantum Polynomial-time'') is a quantum circuit of the form $\mathcal{C} = H^{\otimes n} D H^{\otimes n}$, where $H$ is the Hadamard gate and $D$ is a diagonal matrix produced from $\poly(n)$ diagonal gates. The {\bf IQP sampling problem} is to sample from the distribution $p$ on $n$-bit strings produced by applying $\mathcal{C}$ to the initial state $\ket{0}^{\otimes n}$, then measuring each qubit in the computational basis. (Throughout, $p$ denotes this original noise-free distribution.)

Our first main result is the following:

\begin{thm}[informal]
\label{thm:square}
There is a family of commuting quantum circuits on $n$ qubits where: with high probability, a random circuit picked from the family contains $O(n \log n)$ 2-qubit gates and can be implemented on a 2d square lattice in depth $O(\sqrt{n} \log n)$; and a constant fraction of circuits picked from the family cannot be simulated classically unless the polynomial hierarchy collapses to the third level, assuming a ``sparse'' version of Conjecture \ref{con:isingorig}. Here ``simulate'' means to approximately sample from the output distribution of the circuit, up to $\ell_1$ distance $\epsilon$, for some constant $\epsilon > 0$.
\end{thm}

In the above we use ``2d square lattice'' as shorthand for an architecture consisting of a square lattice of $\sqrt{n} \times \sqrt{n}$ qubits, where all gates are performed across neighbours in the lattice. To prove Theorem \ref{thm:square} we proceed as follows. First, we show that the ``dense'' IQP circuits from~\cite{bremner16}, which contained $O(n^2)$ gates, can be reduced to ``sparse'' circuits of $O(n \log n)$ long-range gates while still likely being hard to simulate. Second, we show that a random circuit of this form can be parallelised to depth $O(\log n)$ with high probability. Third, we apply the results of~\cite{beals13} to show that sorting networks can be used to implement an arbitrary quantum circuit of depth $t$ on a 2d square lattice in depth $D=O(t \sqrt{n})$. (Note that this final circuit is no longer an IQP circuit as it contains SWAP gates.)

While it might seem that sparse IQP sampling is more likely to be classically simulable, it is possible that the converse is true. The complexity-theoretic hardness arguments rely on the conjecture that complex temperature partition functions of the Ising model retain \#P-hardness on random graphs~\cite{bremner16}. It is known that there are a range of related \#P-hard and NP-hard graph problems that admit an efficient approximation for random dense graphs \cite{arora99, alon94}, while retaining their hardness on sparse graphs. However, it should be stressed that there are no known efficient approximation methods for the complex temperature partition functions associated with sparse- and dense-IQP sampling.

It remains to be seen whether a more sparse version of IQP sampling can be devised while retaining its classical hardness. Standard tensor network contraction techniques would allow any output probability of the above circuits on a square lattice to be classically computed in time $O(2^{D\sqrt{n}})$, so achieving a similar hardness result for $D = o(\sqrt{n})$ would violate the counting exponential time hypothesis~\cite{brand16, curticapean15}. The challenge remains to remove a factor of $\log n$ from the depth while maintaining the anticoncentration requirements of \cite{aaronson13,bremner16}. 

It is worth comparing Theorem \ref{thm:square} with results of Brown and Fawzi~\cite{brown12,brown15}. In~\cite{brown15}, these authors show that random noncommutative quantum circuits with $O(n \log^2 n)$ gates are good decouplers (a somewhat similar notion of randomisation), and that such circuits can be parallelised to depth $O(\log^3 n)$ with high probability. Using a sorting network construction, these circuits could be implemented on a 2d square lattice in depth $O(\sqrt{n} \log^3 n)$. Our result thus saves an $O(\log^2 n)$ factor over~\cite{brown15}. One reason for this is that the commutative nature of IQP circuits makes them easier to parallelise. However, in~\cite{brown12}, Brown and Fawzi also study an alternative model for random circuits, where gates are applied at each timestep according to a random perfect matching on the complete graph, and show that this achieves a weaker notion of ``scrambling'' in depth $O(\log n)$. Although it is not clear that this notion in itself would be sufficient for a complexity-theoretic hardness argument, it is thus plausible that our results could be extended to noncommuting circuits. It should also be noted that recent \cite{boixo16} numerical evidence suggests that anticoncentration can be achieved on a square lattice with circuits of depth $O(\sqrt{n})$, potentially allowing for random circuit sampling quantum supremacy experiments.

Next we study the effect of noise on IQP circuits. We consider a very simple noise model: independent depolarising noise applied to every qubit at the end of the circuit. First the IQP circuit is applied to $\ket{0}^{\otimes n}$ as normal; let $\ket{\psi}$ be the resulting state. Then the qubit depolarising channel $\mathcal{D}_\epsilon$ with noise rate $\epsilon$ is applied to each qubit of $\ket{\psi}$. This channel is defined by $\mathcal{D}_\epsilon(\rho) = (1-\epsilon)\rho + \epsilon \frac{I}{2}$ for any mixed state $\rho$ of a single qubit; with probability $1-\epsilon$, the input state is retained, and with probability $\epsilon$, it is discarded and replaced with the maximally mixed state. Finally, each qubit is measured in the computational basis to give a distribution $\widetilde{p}$. (Throughout, $\widetilde{p}$ denotes the distribution created by incorporating some local noise.)

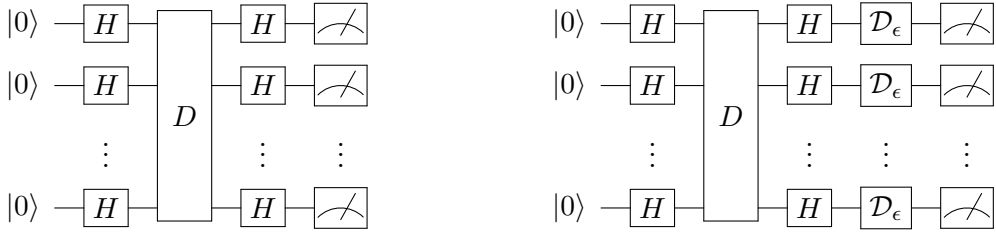
\begin{figure}
\[
\Qcircuit @C=1em @R=.7em {
\lstick{\ket{0}} & \gate{H} & \multigate{3}{D} & \gate{H} & \meter &&&&&&&& \lstick{\ket{0}} & \gate{H} & \multigate{3}{D} & \gate{H} & \gate{\mathcal{D}_\epsilon} & \meter  \\
\lstick{\ket{0}} & \gate{H} & \ghost{D} & \gate{H} & \meter &&&&&&&& \lstick{\ket{0}} & \gate{H} & \ghost{D} & \gate{H} & \gate{\mathcal{D}_\epsilon} &\meter\\
& \push{\vdots} & & \vdots & \vdots & &&&&&&&& \push{\vdots} & & \vdots & \vdots & \vdots \\
\lstick{\ket{0}} & \gate{H} & \ghost{D} & \gate{H} & \meter &&&&&&&& \lstick{\ket{0}} & \gate{H} & \ghost{D} & \gate{H} & \gate{\mathcal{D}_\epsilon} & \meter \\
}
\]
\caption{A standard IQP circuit, and an IQP circuit with depolarising noise. $D$ is a circuit made up of $\poly(n)$ diagonal gates.}
\end{figure}

Note that this model is equivalent to sampling a bit-string $x \in \{0,1\}^n$ according to $p$, then flipping each bit of $x$ with independent probability $\epsilon/2$. To see this, first note that the operation of measuring a qubit in the computational basis commutes with $\mathcal{D}_\epsilon$. If we write $\mathcal{M}(\rho) = \proj{0} \rho \proj{0} + \proj{1} \rho \proj{1}$ for this measurement operation, then
\beas
\mathcal{M}(\mathcal{D}_\epsilon(\rho)) &=& \proj{0} \left((1-\epsilon)\rho + \epsilon \frac{I}{2}\right) \proj{0} + \proj{1} \left((1-\epsilon)\rho + \epsilon \frac{I}{2}\right) \proj{1} = (1-\epsilon) \mathcal{M}(\rho) + \epsilon \frac{I}{2}\\
&=& \mathcal{D}_\epsilon(\mathcal{M}(\rho)).
\eeas
Second, when applied to $\proj{0}$, $\mathcal{D}_\epsilon$ replaces it with the state $(1-\epsilon/2)\proj{0} + (\epsilon/2)\proj{1}$, i.e.\ applies a NOT operation to the state with probability $\epsilon/2$; the same is true when applied to $\proj{1}$.

We also remark that, for IQP circuits, this notion of noise is equivalent to applying depolarising noise to the qubits at the start of the computation. This is because noise at the start of the computation is equivalent to replacing the initial state $\ket{0^n}$ with a state $\ket{y}$ where $y$ is distributed as a noisy version of $0^n$, and $\bracket{x}{\mathcal{C}}{y} = \bracket{x+y}{\mathcal{C}}{0}$.

We first show that if fault-tolerance techniques are not used, then ``most'' IQP circuits can be classically simulated approximately if {\em any} constant amount of noise is applied in this model. The notion of approximate simulation we use is sampling up to accuracy $\delta$ in $\ell_1$ norm, i.e.\ sampling from some distribution $\widetilde{p}'$ such that $\sum_x |\widetilde{p}'_x - \widetilde{p}_x| \le \delta$. (Throughout, $\widetilde{p}'$ denotes any distribution that is close to $\widetilde{p}$ in $\ell_1$ norm.) We will show:

\begin{thm}
\label{thm:main}
Consider a unitary circuit $\mathcal{C} = H^{\otimes n} D H^{\otimes n}$ whose diagonal part $D$ is defined by $\bracket{x}{D}{x} = f(x)$ for some $f: \{0,1\}^n \rightarrow \C$ such that $f(x)$ can be computed in time $\poly(n)$ for any $x$. Let the probability of receiving output $x$ after applying $\mathcal{C}$ to input $\ket{0}^{\otimes n}$ be $p_x$, and assume that $\sum_x p_x^2 \le \alpha 2^{-n}$ for some $\alpha$. Further assume $\mathcal{C}$ experiences independent depolarising noise on each qubit with rate $\epsilon$ as defined above. Then $T$ samples can be generated from a distribution which approximates the noisy output probability distribution up to $\delta$ in $\ell_1$ norm, in time $n^{O(\log(\alpha/\delta)/\epsilon)} + T \poly(n)$.
\end{thm}

The parameter $\alpha$ occurring in Theorem \ref{thm:main} measures how spread out the output probability distribution of $\mathcal{C}$ is. It is shown in~\cite{bremner16} that, for random IQP circuits picked from some natural distributions, the expected value of $\alpha$ is $O(1)$. Hence, for an average IQP circuit picked from one of these distributions, and for fixed $\delta$ and $\epsilon$, the runtime of the classical algorithm is polynomial in $n$. The circuits that were proven hard to simulate in~\cite{bremner16} (assuming some conjectures in complexity theory) have $\alpha = O(1)$. So Theorem \ref{thm:main} shows that precisely those circuits which are hard to simulate in the absence of noise become easy in the presence of noise.

This theorem actually covers cases more general than IQP, since computing $f$ could even require ancilla qubits that are not available in the usual IQP model. Indeed, as well as the application to IQP, the ideas behind Theorem \ref{thm:main} can also be used to show that, in the absence of fault-tolerance, Simon's algorithm~\cite{simon97} can be simulated classically if an arbitrarily small amount of depolarising noise is applied to each qubit. The proof of Theorem \ref{thm:main} uses Fourier analysis over $\Z_2^n$ to show that a noisy output probability distribution $\widetilde{p}$ can be approximated well given the knowledge of only a small number of its Fourier coefficients, because the high-order coefficients are exponentially suppressed by the noise.

Our final result is that this notion of noise can be fought using simple ideas from classical error-correction, while still remaining within the framework of IQP. We show that for any IQP circuit $\mathcal{C}$ on $n$ qubits, we can produce a new IQP circuit $\mathcal{C}'$ on $O(n)$ qubits in polynomial time such that, if depolarising noise is applied to every qubit of the output of $\mathcal{C}'$, we can nevertheless sample from a distribution which is close to $p$ up to arbitrarily small $\ell_1$ distance. This holds for {\em any} noise rate $\epsilon<1$, contrasting with standard fault-tolerance thresholds. (However, the notion of noise here is different and substantially simpler than the usual models.) Crucially, this noise-tolerance can be combined with the notion of approximation used in~\cite{bremner16} to show that, under the same complexity assumptions as~\cite{bremner16}, it is hard for a classical algorithm to approximately sample from the noisy output distribution of $\mathcal{C}'$, up to small $\ell_1$ distance.

\begin{thm}
\label{thm:main2}
Assume either Conjecture \ref{con:isingorig} or Conjecture \ref{con:poly}. Let $\mathcal{C} = H^{\otimes n} D H^{\otimes n}$ be an IQP circuit which experiences independent depolarising noise on each qubit with rate $\epsilon$ as defined above, for some $\epsilon < 1$. Then there exists $\delta > 0$ such that, if there is a polynomial-time classical algorithm which samples from the output probability distribution of all noisy IQP circuits $\mathcal{C}$ of this form up to accuracy $\delta$ in $\ell_1$ norm, the polynomial hierarchy collapses to its third level.
\end{thm}

Local noise more general than that arising from single-qubit depolarising channels may also be dealt with via our method of classical error correction.  Writing $x$ for a sample from the noise-free distribution $p$, and $x+e$ for a sample from $\widetilde{p}$, we see that $e$ is distributed such that $\Pr[e=e'] = (\epsilon/2)^{|e'|}(1-\epsilon/2)^{n-|e'|}$. But in fact we show in Section \ref{sec:ft} that \emph{any} local noise model that makes $e$ overwhelmingly likely to have small Hamming weight would equally well be tolerated by the incorporation of classical error correction.

Thinking of an IQP circuit as a Hamiltonian which is diagonal in the X basis, the error-correction approach we use can be viewed as encoding the terms in the Hamiltonian with a classical error-correcting code. The idea of encoding a Hamiltonian in this way with a classical or quantum code has previously been used to protect adiabatic quantum algorithms against noise (see~\cite{marvian16} and references therein). In the setting of IQP, the analysis becomes particularly clean and simple.


\subsection{Related work and perspective}

{\bf Circuit depth and optimal sparse IQP sampling.} Below we improve on the results of \cite{bremner16} to extend the hardness results of IQP sampling introduced in \cite{bremner16} to sparsely connected circuits. The motivation for this is both theoretical and practical. We want to both improve the likelihood that the hardness conjectures that are made are correct, while also decreasing the physical requirements of the IQP sampling protocol.

The complexity of dense IQP sampling depends on the conjecture that average-case complexity of complex temperature Ising model partition functions over dense graphs is \#P-hard. That is, that the average and worst case complexities coincide for a large fraction of randomly chosen graphs. It is natural to assume that the complexity of combinatorial problems on graphs increases with the density of the graph instances, however this is known to not always be the case. A number of key combinatorial problems that do not generally admit (classical) polynomial time approximation schemes do admit such approximations on dense instances \cite{arora99, alon94}. While these results do not hold for the hardness conjectures made in \cite{bremner16}, they are a clear incentive to determine to what extent the IQP sampling argument can be applied to Ising models on sparse graphs.

In \cite{bremner11} it was shown that IQP sampling, up to relative errors, could not be efficiently performed classically without a collapse in the polynomial hierarchy. It was also noted in \cite{bremner11} that this result still holds for IQP circuits with only nearest neighbour gates arranged on a 2d lattice. If this result could be extended to apply to classical simulations that are reasonably close in total variation distance it would be a massive improvement over the results of \cite{bremner16}. Such circuits could be implemented in constant depth with nearest neighbour interactions, suggesting an exceptional target for quantum supremacy experiments. Unfortunately, the techniques used in~\cite{aaronson13, bremner16} to argue for hardness of simulation up to small total variation distance require the output probability distribution of the circuit to ``anticoncentrate" with high probability, i.e.\ to be spread out, and it does not appear that IQP circuits on a square lattice display sufficient anticoncentration for these techniques to be applicable.

Therefore, Theorem \ref{thm:square} is proven by showing that sparse circuits of $O(n \log n)$ long-range gates anticoncentrate, and then showing that such circuits can be implemented on a 2d square lattice of size $\sqrt{n}\times\sqrt{n}$ in depth $O(\sqrt{n}\log n)$. Recent results relating lower bounds for computing sparse Tutte polynomials to the exponential time hypothesis demonstrate that this is likely close to the optimal depth. Last year it was shown that precise evaluations of Tutte polynomials on sparse graphs cannot be performed in time $\exp(o(n))$ without a violation of the counting equivalent of the exponential time hypothesis \cite{brand16, curticapean15}. That is, if there were a sub-exponential runntime for Tutte polynomials on sparse graphs at all \#P-hard points, then key NP-hard problems such as 3SAT could also be solved in sub-exponential time. The Ising models studied here are examples of \#P-hard points of complex-variable Tutte polynomials \cite{shepherd10}. Tensor network contraction techniques can be used to show that any output probability of any quantum circuit of depth $D$ implemented on a 2d square lattice can be precisely evaluated classically in time $O(2^{D\sqrt{n}})$~\cite{markov08}, suggesting that if it were possible to implement arbitrary sparse IQP circuits in depth $o(\sqrt{n})$ then we are likely to violate the exponential time hypothesis. 

The question remains if it is possible to identify a sampling problem that matches the $O(\sqrt{n})$ depth bound while also remaining classically difficult to simulate. Recent numerical studies indicate that it might be possible to find random circuits that are drawn from universal gate sets that anticoncentrate with depth $O(\sqrt{n})$ on a 2d square lattice \cite{boixo16}. However, an analytic proof that this is possible remains an open question. Finally, it should be noted that a recent paper has suggested that hardness of approximate IQP sampling up to small total variation distance could be proven for IQP circuits that do not necessarily satisfy the anticoncentration property~\cite{gao17}. In this work, the anticoncentration property is replaced with the assumption that most amplitudes corresponding to the results of measurements applied to a 2D ``brickwork'' state, which is universal for measurement-based quantum computing, are hard to approximately compute. The approach of~\cite{gao17} leads to a lower-depth circuit than ours, but with a polynomial increase in the number of qubits; and as the hardness assumption used is somewhat different, the results are not directly comparable with ours. Subsequently to the first version of this paper, Bermejo-Vega et al.~\cite{bermejovega17} have described several other constant-depth architectures which have a similar polynomial increase in size, but whose hardness is based on conjectures closer to those we use here.

{\bf Hardness results for noisy IQP.} It was recently shown by Fujii and Tamate~\cite{fujii16}, using the theory of quantum fault-tolerance, that the distributions produced by IQP circuits are classically hard to simulate, even under a small amount of noise. That is, a quantum channel $\mathcal{N}$ is applied to each qubit of the output state such that $\|\mathcal{N} - \operatorname{id}\|_\diamond \le \epsilon$ for a sufficiently small constant $\epsilon$. Fujii and Tamate~\cite{fujii16} show that the resulting distribution cannot be sampled from classically unless the polynomial hierarchy collapses to the third level. Theorem \ref{thm:main} may appear to be in conflict with this result; however, this is not the case. Fujii and Tamate's result shows that it is classically hard to sample from the noisy output distribution $\widetilde{p}$ of arbitrary IQP circuits up to a small relative error in each probability. Theorem \ref{thm:main} shows that for random IQP circuits these distributions can nevertheless be sampled from approximately, if the notion of approximation used is $\ell_1$ distance.

Note that the notion of multiplicative approximation used in~\cite{fujii16} (and also~\cite{bremner11}) is a very strong one: for example, if any of the output probabilities are 0, this 0 must be reproduced exactly in the sampled distribution. By contrast, the $\ell_1$ distance is a physically realistic measure of distance. For example, if $\| \widetilde{p}' - \widetilde{p} \|_1 \le \epsilon$, $\Omega(1/\epsilon)$ samples are required to distinguish between $\widetilde{p}'$ and $\widetilde{p}$. The framework of relative-error approximation appears naturally in~\cite{fujii16} because that work applies standard quantum fault-tolerance within a postselected version of IQP, and approximation up to small $\ell_1$ error does not combine well with postselection. In order to show that noisy versions of IQP circuits are hard to simulate up to small $\ell_1$ error, it appears necessary to use a notion of fault-tolerance which is itself native to IQP, as we do in Theorem \ref{thm:main2}.

It was shown in~\cite{bremner16} that classical sampling from the output distribution of random IQP circuits up to $\ell_1$ distance smaller than a universal constant $c$ is hard, assuming either of two reasonable average-case hardness conjectures. Again, this is not in conflict with the classical simulation results given here: applying noise to the output distributions of the circuits which are hard to simulate in~\cite{bremner16} could change them dramatically. Indeed, if depolarising noise with rate $\epsilon$ is applied to each qubit of an $n$-qubit quantum state, the distance between the resulting state and the original state could be as high as $\Omega(n \epsilon)$. So a constant amount of noise on each qubit is easily sufficient to leave the regime which was shown to be hard in~\cite{bremner16}.

The results obtained here are compared with previously known results in Table \ref{tab:compare}.

\begin{table}
\begin{tabular}{|c|c|c|}
\hline & Multiplicative approximation & Additive approximation \\
\hline \hline Noise-free & Hard (if PH does not collapse)~\cite{bremner11} & Hard (w/ stronger complexity assumptions)~\cite{bremner16}\\
\hline Noisy & Hard (if PH does not collapse)~\cite{fujii16} & Hard (general circuits, similar assumptions) / \\
 &  & polynomial-time (random circuits) \\
\hline
\end{tabular}
\caption{Comparison of hardness results for simulating IQP circuits classically. ``Multiplicative approximation'' means the task of sampling from the output distribution up to small relative error in each probability; ``additive approximation'' is the task of sampling from the output distribution up to small $\ell_1$ distance. ``Noisy'' means depolarising noise with rate $\epsilon$ applied to each qubit of the output state, for some small fixed $\epsilon > 0$. PH is short for ``polynomial hierarchy''.}
\label{tab:compare}
\end{table}

{\bf Classical simulation of general quantum circuits.} The theory of quantum fault-tolerance states that there is a constant noise threshold below which universal quantum computation is possible. A number of authors have provided converses to this, in a variety of different models~\cite{razborov04,virmani05,buhrman06a,kempe08}. These works show that, if a quantum circuit experiences sufficient noise, either it is simulable classically, or its output is essentially independent of its input. Perhaps the most relevant of these results to the case of IQP circuits is that of Razborov~\cite{razborov04}, which considers arbitrary quantum circuits containing gates of fan-in at most $k$, and a model where depolarising noise with rate larger than $1-1/k$ is applied to each qubit after each layer of gates in the circuit. It is shown in~\cite{razborov04} that, after $O(\log n)$ layers of gates, the output state of $n$ qubits produced by the circuit essentially does not depend on the input to the circuit. Theorem \ref{thm:main} is a rare case where there is no threshold noise rate: there is a classical algorithm which approximately samples from the output distribution for {\em any} noise rate $>0$. This does not contradict standard fault-tolerance results, because fault-tolerance techniques have not been applied to the IQP circuits which are classically simulable.

{\bf Boson sampling.} The {\em boson sampling} problem of Aaronson and Arkhipov~\cite{aaronson13} is defined as follows. For an $m \times n$ column-orthonormal matrix $U$, approximately sample from the distribution $\mathcal{D}_{\operatorname{bs}}$ on sequences $S = (s_1,\dots,s_m)$, where the $s_i$ are nonnegative integers which sum to $n$, given by
\be \label{eq:bs} \Pr[S] = \frac{|\operatorname{Perm}(U_S)|^2}{s_1!\dots s_m!} \ee
where $U_S$ is the $n \times n$ submatrix of $U$ containing $s_i$ copies of the $i$'th row of $U$, for all $i = 1,\dots m$, and $\operatorname{perm}(U_S)$ is the permanent of $U_S$. 

Kalai and Kindler have given evidence that suggests that, for small errors in the matrix $U$, boson sampling is classically simulable~\cite{kalai14} (see also~\cite{kalai16}, and~\cite{rahimikeshari16} for a recent study of more physically-motivated noise models). To be precise, they show the following. Let $X$ and $U$ be random Gaussian matrices ($n \times n$ matrices whose entries are picked from a normalised Gaussian distribution), and set $Y = \sqrt{1-\epsilon} X + \sqrt{\epsilon} U$ for some $\epsilon = \omega(1/n)$. Write $f(X) = |\operatorname{perm}(X)|^2$, $g(X) = \E[|\operatorname{perm}(Y)|^2 | X ]$. Then, for any $d \gg 1/\epsilon$, there is a degree-$d$ polynomial $p$ such that $\|p(X) - g(X)\|_2^2 = o(\|g\|_2^2)$, and $p$ can be efficiently approximated classically to within a constant.

It was also shown by Leverrier and Garc\'ia-Patr\'on~\cite{leverrier15}, and independently Kalai and Kindler~\cite{kalai14}, that, for considerably smaller levels of imperfection (e.g.~$\epsilon \gg 1/n$), the output of the boson sampling circuit is far from the intended output. Note that, in the intermediate regime $\epsilon = o(1)$, $\epsilon \gg 1/n$, the output of the circuit could still be hard to approximate while being far from the intended output. On the other hand, it was shown by Arkhipov~\cite{arkhipov15} (see also~\cite{shchesnovich15}) that if $\epsilon = o(1/n^2)$, the $\ell_1$ distance between the noisy distribution and the original distribution is $o(1)$.

As discussed in~\cite{kalai14}, the results of Kalai and Kindler do not quite imply that the boson sampling problem as described in~\cite{aaronson13} can be solved classically with a sufficiently large (but constant) amount of noise. The results of~\cite{kalai14} cannot simply be averaged over $S$ to obtain a similar low-degree polynomial approximation to $\mathcal{D}_{\operatorname{bs}}$, as they do not take the normalisation term in (\ref{eq:bs}) into account, nor the possibility of repeated columns in $S$. Nevertheless, they provided the first rigorous evidence that boson sampling in the presence of noise could be classically simulable. Based on this evidence, it was conjectured in~\cite{kalai16} that ``small noisy quantum circuits and other similar quantum systems'' could be approximated by low-degree polynomials. The present work proves this conjecture for the first time for a nontrivial class of quantum circuits, using similar ``noise sensitivity'' ideas to~\cite{kalai14}.

{\bf The noise model.} Noise models are deeply specific to any given implementation of a quantum computation. The noise model considered in this paper is relatively simple, where a perfect implementation of the desired circuit is followed by independent depolarising noise on each qubit in the circuit. As this is at the end of the circuit, it results in independent bitflip noise on each qubit. 

Despite the simplicity, it is a reasonable testbed for several physically relevant scenarios. A common noise model, and the model in which the fault-tolerance threshold theorem is proven, would have noise applied before and after every gate in the circuit, rather than just at the beginning or end as here. If the intermediate errors are dephasing errors, then this scenario is equivalent to the model studied in this paper. This follows from two facts. Firstly, sequential dephasing maps compose into another dephasing map, albeit one with a higher probability of error. The second key feature is that dephasing maps commute with the diagonal gates in an IQP circuit. These can be ``commuted through'' the Hadamard gates to produce bitflip channels. Finally, dephasing at the beginning and end of the circuit is not observable.

The results of Section \ref{sec:ft} demonstrate that IQP circuits can be made fault tolerant to dephasing errors using only marginally larger IQP circuits. However it is not clear that more general noise models, for example those allowing for depolarising errors between gates, can be made correctable within IQP (unless of course the entire IQP circuit is trivially regarded as a single gate acting on the whole system). For example, consider a circuit made up of CZ gates, each of which has depolarising noise applied to both of its qubits before and after the gate (call these NCZ gates). Then NCZ gates do not commute with one another, even when applied to the initial state $\ket{+}^{\otimes n}$. This opens up the intriguing possibility that noise could actually {\em increase} the power of IQP circuits, by allowing them to sample from otherwise inaccessible distributions.

{\bf Perspective on these results and quantum supremacy.} We feel that our results highlight the challenges for quantum supremacy experiments in the presence of noise, and also the challenges for skeptics attempting to prove that quantum supremacy is impossible. In the case of IQP circuits that are apparently hard to simulate classically, then if no fault-tolerance is used, the circuits can be simulated in polynomial time if there is a very small amount of noise. On the other hand, correcting noise of a rather natural form can be achieved using only classical ideas, with no need for the full machinery of quantum fault-tolerance, and only a small increase in the size of the circuit. The setting of IQP serves as a simple laboratory in which to explore these issues, which we expect will also apply to other proposed experiments. Another important challenge, as for all sampling problems, is to find a simple method for verifying that an experimental implementation of IQP sampling has been correctly implemented. An IQP verification procedure was proposed in~\cite{hangleiter17}, but this requires the preparation of states going beyond the IQP model.

We finally remark that, although our classical simulation of noisy IQP circuits runs in polynomial time, it is not remotely efficient in practice for reasonable noise rates (e.g.\ $\epsilon \approx 0.01$), as the runtime exponent depends linearly on $1/\epsilon$. A suitable experiment could still demonstrate quantum supremacy over this algorithm even without an exponential separation being possible.


\section{Sparse IQP circuits}
\label{sec:sparseIQP}
In this section we discuss how to parallelise IQP circuits and implement them on a square lattice. The first step is to replace the ``dense'' IQP circuits studied in~\cite{bremner16} with a sparser type of circuit, which will be easier to parallelise. We consider the following method of choosing the diagonal part of a random IQP circuit $\mathcal{C}$ on $n$ qubits:

\begin{itemize}
\item For each possible choice of a pair $(i,j)$ of distinct qubits, include a gate across those qubits with probability $p = \gamma (\ln n) / n$, for some fixed $\gamma > 0$.
\item Each 2-qubit gate is picked uniformly at random from the set $\{\diag(1,1,1,\omega^k): k \in \{0,\dots,3\}\}$, where $\omega = i$.
\item Each qubit has a 1-qubit gate acting on it, which is picked uniformly at random from the set $\{\diag(1,\zeta^k): k \in \{0,\dots,7\}\}$, where $\zeta = e^{\pi i/4}$.
\end{itemize}

Call an IQP circuit picked from this distribution {\em sparse}. Sparse IQP circuits contain $O(n \log n)$ gates with high probability and are a variant of the ``Ising-like'' class of IQP circuits considered in~\cite{bremner16}. Indeed, for any circuit $\mathcal{C}$ of the above form, we have
\[ \bracket{0}{\mathcal{C}}{0} = \sum_{x \in \{0,1\}^n} \zeta^{\sum_{i<j} w_{ij} x_i x_j + \sum_k v_k x_k} \]
for some integer weights $w_{ij}$, $v_k$: this is easily seen to correspond to an Ising model partition function $Z_{\mathcal{C}}(\zeta)$ (cf.~(\ref{eq:isingorig})). We will need the following key technical lemma, a sparse counterpart of anticoncentration results proven in~\cite{bremner16}.

\begin{lem}
\label{lem:anticonc}
Let $\mathcal{C}$ be a random sparse IQP circuit. Then $\E_{\mathcal{C}}[|\bracket{0}{\mathcal{C}}{0}|^2] = 2^{-n}$ and, for a large enough constant $\gamma$, $\E_{\mathcal{C}}[|\bracket{0}{\mathcal{C}}{0}|^4] \le 5 \cdot 2^{-2n}$.
\end{lem}

We prove Lemma \ref{lem:anticonc} in Appendix \ref{app:anticonc}. By the Paley-Zygmund inequality, which states that $\Pr[R \ge \alpha\, \E[R]] \ge (1-\alpha)^2 \E[R]^2/\E[R^2]$ for any random variable $R$ with finite variance and any $0 < \alpha < 1$, we have that, for a large enough constant $\gamma$, $\Pr[ |\bracket{0}{\mathcal{C}}{0}|^2 \ge \alpha\cdot 2^{-n} ] \ge (1-\alpha)^2/5$. We use this within the following result from~\cite{bremner16} (slightly rephrased):

\begin{cor}
\label{cor:addapprox}
Let $\mathcal{F}$ be a family of IQP circuits on $n$ qubits. Pick a random circuit $\mathcal{C}$ by choosing a circuit from $\mathcal{F}$ at random according to some distribution, then appending X gates on a uniformly random subset of the qubits. Assume that there exist universal constants $\alpha,\beta>0$ such that $\Pr[|\bracket{0}{\mathcal{C}}{0}|^2 \ge \alpha\cdot2^{-n}] \ge \beta$. Further assume there exists a classical polynomial-time algorithm $\mathcal{A}$ which, for any IQP circuit $\mathcal{C}'$ of this form, can sample from a probability distribution which approximates the output probability distribution of $\mathcal{C}'$ up to additive error $\epsilon = \alpha \beta / 8$ in $\ell_1$ norm. Then there is a $\FBPP^{\NP}$ algorithm which, given access to $\mathcal{A}$, approximates $|\bracket{0}{\mathcal{C}}{0}|^2$ up to relative error $1/4 + o(1)$ on at least a $\beta/2$ fraction of circuits $\mathcal{C}$.
\end{cor}

In this corollary, $\FBPP^{\NP}$ is the complexity class corresponding to polynomial-time classical randomised computation, equipped with an oracle to solve NP-complete problems. By inserting the parameters from Lemma \ref{lem:anticonc}, we see that there are universal constants $0 < \epsilon, c < 1$ such that the following holds: If there is a classical algorithm $\mathcal{A}$ which can sample from a probability distribution approximating the output probability distribution of any such circuit $\mathcal{C}'$ up to additive error $\epsilon$ in $\ell_1$ norm, then there is a $\FBPP^{\NP}$ algorithm which, given access to $\mathcal{A}$, approximates $|\bracket{0}{\mathcal{C}}{0}|^2$ up to relative error $1/4 + o(1)$ on at least a $c$ fraction of sparse IQP circuits $\mathcal{C}$. Note that, by a union bound, we can weaken the requirement that the algorithm $\mathcal{A}$ works for {\em all} such circuits $\mathcal{C}'$ to the requirement that it works for a large constant fraction of them, at the expense of reducing the constant $c$.

We conjecture that this latter problem is \#P-hard. This corresponds to approximating the partition function of the Ising model up to small relative error, for random graphs that are relatively sparse (yet still connected with high probability), which is a similar hardness assumption to one considered in~\cite{bremner16}. If this conjecture holds, then the existence of such a classical sampler would imply collapse of the polynomial hierarchy~\cite{toda91}, a complexity-theoretic consequence considered very unlikely; see~\cite{bremner16} for more.

Thus the conjecture that we make is as follows (cf.\ Conjecture \ref{con:isingorig}), where we choose $\alpha = 1/2$ in Corollary \ref{cor:addapprox} for concreteness, giving $\beta = 1/40$, $c=1/80$:

\begin{con}
\label{con:ising}
There is a universal constant $c < 1/80$ such that it is \#P-hard to approximate $|Z_{\mathcal{C}}(\zeta)|^2$ up to relative error $1/4 + o(1)$ for an arbitrary $c$ fraction of instances $\mathcal{C}$ picked from the above distribution.
\end{con}

It should be noted that finding such a relative-error approximation to $|Z_{\mathcal{C}}(\zeta)|^2$ is \#P-hard in the worst case even for constant-depth IQP circuits. Note that the precise value of $c$ is not very significant. The decrease in the bound on $c$ compared with Conjecture \ref{con:isingorig} is because the constant in Lemma \ref{lem:anticonc} is somewhat larger than in the equivalent result in~\cite{bremner16}.


\subsection{Parallelising IQP circuits}

Next we show that sparse IQP circuits can be parallelised efficiently, assuming that long-range interactions are allowed. An arbitrary IQP circuit whose gates act on at most 2 qubits can be implemented by first implementing the 2-qubit gates (combining multiple gates acting on the same qubits into one gate), and then implementing the 1-qubit gates in one additional layer. So consider an IQP circuit $\mathcal{C}$ on $n$ qubits, where each gate acts on 2 qubits, and such that there is at most one gate acting across each pair of qubits.

$\mathcal{C}$ can be implemented in depth $t$ if the gates can be partitioned into $t$ sets such that, within each set, no pair of gates ``collide'' (act on the same qubit). Let $G_{\mathcal{C}}$ be the corresponding graph on $n$ vertices which has an edge between vertices $i$ and $j$ if $\mathcal{C}$ has a gate between qubits $i$ and $j$. Then such a partition of $\mathcal{C}$ is equivalent to colouring the edges of $G_{\mathcal{C}}$ with $t$ colours such that no pair of edges incident to the same vertex share the same colour. Vizing's theorem~\cite{diestel10} states that any graph $G$ has a proper edge-colouring of this form with at most $\Delta(G)+1$ colours, where $\Delta(G)$ is the maximal degree of a vertex of $G$. So all that remains is to bound $\Delta(G_{\mathcal{C}})$ for a random sparse IQP circuit $\mathcal{C}$.

This is equivalent to bounding $\Delta(G)$ for a random graph $G$ such that each edge is present with probability $p = \gamma (\ln n) / n$. The maximum degree of random graphs has been studied in detail (see e.g.~\cite{bollobas80}); here we give an elementary upper bound.

\begin{lem}
Let $G$ be a random graph where each edge is present with probability $p = \gamma (\ln n) / n$. Then $\Pr[\Delta(G) \ge 2 \gamma (\ln n)] \le n^{1-\gamma/4}$.
\end{lem}

\begin{proof}
By a union bound, for any $d$, $\Pr[\Delta(G) \ge d] \le n \Pr[\text{deg}(v_1) \ge d]$, where $\text{deg}(v_1)$ is the degree of a fixed vertex $v_1$. The degree of $v_1$ is the number of edges incident to $v_1$; each edge is present with probability $p$; so by a Chernoff bound argument~\cite{dubhashi09}
\[ \Pr[\text{deg}(v_1) \ge 2 \gamma (\ln n)] \le e^{-\gamma (\ln n) / 4} = n^{-\gamma / 4}. \]
The claim follows.
\end{proof}

So, for a large enough constant $\gamma$, the probability that $\Delta(G) \ge 2 \gamma (\ln n)$ is negligible. Note that, in this regime, with high probability $G$ is connected and has maximal treewidth, implying that it is not obvious how to simulate $\mathcal{C}$ classically using tensor-contraction techniques~\cite{markov08}.

We can therefore parallelise a random IQP circuit containing $O(n\log n)$ gates to depth $O(\log n)$, which is optimal. It is worth comparing this to the  bounds obtained in~\cite{brown15} for parallelising general quantum circuits. There it was shown that a random circuit of depth $t$ can be parallelised to depth $O(t (\log n)/n)$ with high probability. Here we have removed a log factor by taking advantage of our ability to commute gates through each other.


\subsection{Sorting networks}

We finally show how to implement sparse IQP circuits depth-efficiently on a 2d square lattice. Consider an arbitrary quantum circuit $\mathcal{C}$ on $n$ qubits of depth $t$. We would like to implement $\mathcal{C}$ on a 2d square lattice of $\sqrt{n} \times \sqrt{n}$ qubits. It is known~\cite{beals13} that, for any geometric arrangement of $n$ qubits, {\em sorting networks} on that geometry correspond to efficient implementations of quantum circuits in that geometry. A sorting network on $n$ elements is a kind of circuit on $n$ lines, where each line is thought to carry an integer, and each gate across two lines is a comparator which swaps the two integers if they are out of order. Sorting networks are designed such that, at the end of the sorting network, any input sequence will have been sorted into ascending order. The aim is to minimise the depth of the network, while possibly obeying geometric constraints (such as comparisons needing to occur across nearest neighbours in some lattice geometry).

We briefly sketch the argument that sorting networks give efficient implementations of circuits on particular geometries~\cite{beals13}. Imagine we have a sequence of non-nearest-neigbour 2-qubit gates to apply in parallel, each (necessarily) acting on distinct qubits, but that we are only allowed to perform nearest-neighour gates (in some geometry). To perform this sequence, it is sufficient to rearrange the qubits such that each pair across which we want to apply a gate is adjacent, then perform the gates (in parallel), then rearrange the qubits to their original order. To do this, we would like to perform a certain permutation of the qubits using only SWAP gates, where each SWAP gate acts across nearest neighbours.
 
This is almost exactly what sorting networks achieve. Each gate in a sorting network can be thought of as a controlled-SWAP, where the values in the two lines are swapped if they are in the incorrect order. To produce a circuit of SWAPs from a sorting network to achieve a desired permutation $\sigma$, we can feed in the sequence $\sigma^{-1}(1),\dots,\sigma^{-1}(n)$ to the network. Whenever a gate is applied to two integers which are currently out of order, we represent it in the circuit by a SWAP gate; otherwise, we do not include it. Assuming that the sorting network works correctly, it will map $\sigma^{-1}(1),\dots,\sigma^{-1}(n)$ to $1,\dots,n$, or in other words will perform the permutation $\sigma$. Any geometric constraints obeyed by the comparators in the sorting network will also be obeyed by the network of SWAPs.

It was shown in~\cite{schnorr86} that there exists a sorting network on a 2d $\sqrt{n} \times \sqrt{n}$ lattice which has depth $3\sqrt{n} + o(\sqrt{n})$; this is close to optimal by diameter arguments. Therefore, any quantum circuit of depth $t$ on $n$ qubits can be implemented on a 2d square lattice of $\sqrt{n} \times \sqrt{n}$ qubits in depth $O(t \sqrt{n})$. Putting all the above pieces together, we have completed the proof of Theorem \ref{thm:square}: there is a family of quantum circuits on $n$ qubits where with high probability a circuit picked from the family contains $O(n \log n)$ 2-qubit commuting gates and can be implemented on a 2d square lattice in depth $O(\sqrt{n} \log n)$; and a constant fraction of circuits picked from the family are hard to simulate classically, assuming similar conjectures to~\cite{bremner16}. Restating Theorem \ref{thm:square} more formally:

\begin{repthm}{thm:square}
Assume Conjecture \ref{con:ising}. Then there is a distribution $\mathcal{D}$ on the set of commuting quantum circuits on $n$ qubits and universal constants $q,\epsilon > 0$ such that: with high probability, a circuit picked from $\mathcal{D}$ contains $O(n \log n)$ 2-qubit commuting gates and can be implemented on a 2d square lattice in depth $O(\sqrt{n} \log n)$; and a $q$ fraction of circuits picked from $\mathcal{D}$ cannot be simulated classically unless the polynomial hierarchy collapses to the third level. Here ``simulate'' means to approximately sample from the output distribution of the circuit, up to $\ell_1$ distance $\epsilon$.
\end{repthm}


\section{Approximating the output probability distribution of noisy IQP circuits}

We now turn to giving a classical algorithm for approximately simulating noisy IQP circuits. We prove that, in many cases, noisy probability distributions $\widetilde{p}$ produced by IQP circuits are approximately classically simulable (Theorem \ref{thm:main}) by showing the following, for any fixed $\delta > 0$:
\begin{enumerate}
\item We can calculate a description of a function $\widetilde{q}$ which approximates $\widetilde{p}$ up to $\ell_1$ error $\delta$, and which has only $\poly(n)$ Fourier coefficients over $\Z_2^n$.

\item We can calculate all marginals of the function $\widetilde{q}$ exactly and efficiently.

\item This enables us to sample from a probability distribution $\widetilde{p}'$ which approximates $\widetilde{p}$ up to $\ell_1$ error $O(\delta)$.
\end{enumerate}
In order to show all these things, we will use some basic ideas from Fourier analysis of boolean functions~\cite{odonnell14}. Any function $f:\{0,1\}^n \rightarrow \C$ can be expanded in terms of the functions $\chi_s(x) = (-1)^{s \cdot x} = (-1)^{\sum_i s_i x_i}$ as
\[ f = \sum_{s \in \{0,1\}^n} \hat{f}(s) \chi_s; \]
the values $\hat{f}(s)$ are called the Fourier coefficients of $f$. It is easy to show that
\[ \hat{f}(s) = \frac{1}{2^n} \sum_{x \in \{0,1\}^n} f(x) (-1)^{s \cdot x}. \]
Fourier analysis is important in the study of IQP because the model can be understood as sampling from the Fourier spectrum of a function $f(x) = \braket{x|D|x}$; indeed, the probability of receiving outcome $s$ when measuring at the end of the circuit is precisely $|\hat{f}(s)|^2$ when noise is absent.

Fourier analysis is also useful to understand the effect of noise. Recall from the introduction that the depolarising noise applied at the end of the circuit is equivalent to applying noise to the output probability distribution $p$ to give a new distribution $\widetilde{p}$. The noise operation applied is precisely the binary symmetric channel, also known simply as the ``noise operator'' for functions on the boolean cube. We denote this classical noise operation $\mathcal{N}_\epsilon$. The Fourier coefficients of the resulting distribution behave nicely under this noise~\cite{odonnell14}:
\[ \widehat{\widetilde{p}}(s) = (1-\epsilon)^{|s|}\widehat{p}(s) \]
for all $s \in \{0,1\}^n$, where $|s|$ is the Hamming weight of $s$.

\subsection{The IQP simulation algorithm}

We first show how to determine a function $\widetilde{q}$ approximating the noisy output distribution $\widetilde{p}$ up to $\ell_1$ error $\delta$, for arbitrary $\delta > 0$. Imagine we know approximations $\widehat{p}'(s)$ to the Fourier coefficients of $p$ for $|s| \le \ell$, for some integer $\ell$, such that $|\widehat{p}'(s) - \widehat{p}(s)| \le \gamma 2^{-n}$ for some $\gamma$. Then our approximation is defined by $\widehat{\widetilde{q}}(s) = (1-\epsilon)^{|s|} \widehat{p}'(s)$ for $|s| \le \ell$, and $\widehat{\widetilde{q}}(s) = 0$ for $|s| > \ell$. So, bounding the $\ell_1$ norm by the $\ell_2$ norm and using Parseval's equality, we have
\beas
\| \widetilde{q} - \widetilde{p} \|_1^2 &\le& 2^n \sum_{x \in \{0,1\}^n} (\widetilde{q}_x - \widetilde{p}_x)^2\\
&=& 2^{2n} \sum_{s \in \{0,1\}^n} (\widehat{\widetilde{q}}(s) - \widehat{\widetilde{p}}(s))^2\\
&=& 2^{2n} \left( \sum_{s,|s| \le \ell} (1-\epsilon)^{2|s|}(\widehat{p}'(s) - \widehat{p}(s))^2 + \sum_{s,|s| > \ell} (1-\epsilon)^{2|s|}\widehat{p}(s)^2 \right)
\eeas
and hence
\beas
\| \widetilde{q} - \widetilde{p} \|_1^2 &\le& \gamma^2(n^\ell+1) + 2^{2n} (1-\epsilon)^{2\ell} \sum_{s \in \{0,1\}^n} \widehat{p}(s)^2\\
&=& \gamma^2(n^\ell+1) + 2^n (1-\epsilon)^{2\ell} \sum_{x \in \{0,1\}^n} p_x^2,
\eeas
where we use $|\{s: |s| \le \ell \}| = \sum_{k=0}^\ell \binom{n}{k} \le n^\ell + 1$. Now assume that $\sum_{x \in \{0,1\}^n} p_x^2 \le \alpha 2^{-n}$ for some $\alpha$. For random IQP circuits, for example, we have $\alpha=O(1)$ with high probability~\cite{bremner16}. Then we have
\[ \| \widetilde{q} - \widetilde{p} \|_1 \le \sqrt{\gamma^2(n^\ell + 1) + \alpha (1-\epsilon)^{2\ell}} \le \sqrt{\gamma^2(n^\ell + 1) + \alpha e^{-2\epsilon \ell}}, \]
so in order to approximate $\widetilde{p}$ up to accuracy $\delta$ in $\ell_1$ norm, it is sufficient to take $\ell = O(\log (\alpha/\delta)/\epsilon)$, $\gamma = O(\delta n^{-\ell/2})$. This corresponds to approximating $n^{O(\log(\alpha/\delta)/\epsilon)}$ Fourier coefficients of $p$ up to accuracy $O(\delta n^{-O(\log(\alpha/\delta)/\epsilon)} 2^{-n})$.

To see that we can approximate these coefficients efficiently, observe that there is a nice expression for them when $p$ is the output probability distribution of an IQP circuit defined by a diagonal matrix $D$, where $\bracket{x}{D}{x} = f(x)$ for some $f:\{0,1\}^n \rightarrow \C$:
%
\beas
\label{eq:conv} \widehat{p}(s) &=& \frac{1}{2^n} \sum_{x \in \{0,1\}^n} p(x) (-1)^{s \cdot x}\\
&=& \frac{1}{2^n} \sum_{x \in \{0,1\}^n} \left| \frac{1}{2^n} \sum_{y \in \{0,1\}^n} f(y) (-1)^{x \cdot y} \right|^2 (-1)^{s \cdot x}\\
&=& \frac{1}{2^{3n}} \sum_{x,y,z \in \{0,1\}^n}  \overline{f(y)}f(z) (-1)^{x \cdot (s+y+z)}\\
&=& \frac{1}{2^{2n}} \sum_{y \in \{0,1\}^n} \overline{f(y)} f(y+s),
\eeas
where $\bar{\cdot}$ denotes complex conjugation\footnote{This can also be seen immediately by observing that the Fourier transform changes multiplication into convolution.}. For any $\eta > 0$, it follows from standard Chernoff bound arguments~\cite{dubhashi09} that we can approximate $2^{-n} \sum_{x \in \{0,1\}^n} \overline{f(x)} f(x + s) = 2^n \widehat{p}(s)$ up to additive error $\eta$ with failure probability $1/3$ using $O(1/\eta^2)$ evaluations of $f$, by simply picking $O(1/\eta^2)$ random values $x \in \{0,1\}^n$, computing $\overline{f(x)} f(x + s)$ and taking the average. Taking the median of $O(\log 1/\zeta)$ repetitions of this procedure reduces the probability of failure to $\zeta$, for any $\zeta >0$. Thus we can approximate $\widehat{p}(s)$ up to additive error $\delta n^{-O(\log(\alpha/\delta)/\epsilon)} 2^{-n}$ with failure probability $1/\poly(n)$ by evaluating $f$ $n^{O(\log(\alpha/\delta)/\epsilon)}/\delta^2$ times. Each such evaluation can be performed in polynomial time. So all of the required coefficients can be approximated up to additive error $\delta 2^{-n}$ in time $n^{O(\log(\alpha/\delta)/\epsilon)}$, with failure probability $o(1)$.

Next, we show that, for any $\widetilde{q}$, knowledge of the Fourier coefficients of $\widetilde{q}$ implies that we can compute its marginals efficiently (see~\cite{shepherd10} for a related discussion). Note that $\widetilde{q}$ is not necessarily a probability distribution: i.e.\ it may take negative values and not sum to 1. Let $x_{1\dots k}$ denote the string consisting of the first $k$ bits of $x$. Assume that $\widetilde{q}$ has $N$ nonzero Fourier coefficients and consider the sum $S_y := \sum_{x,x_{1\dots k} = y} \widetilde{q}(x)$ for each $k \in \{0,\dots,n\}$ and each $y \in \{0,1\}^k$, where for $k=0$ we consider the empty string $y = \emptyset$ and let $S_\emptyset = \sum_x \widetilde{q}(x)$. Then 
\beas
S_y &=& \sum_{x,x_{1\dots k} = y} \widetilde{q}(x)\\
&=& \sum_{x,x_{1\dots k} = y} \sum_{s \in \{0,1\}^n} (-1)^{x \cdot s} \widehat{\widetilde{q}}(s)\\
&=& \sum_{s \in \{0,1\}^n} \widehat{\widetilde{q}}(s) \sum_{x,x_{1\dots k} = y} (-1)^{x \cdot s}\\
\eeas
\beas
S_y &=& \sum_{s \in \{0,1\}^n} \widehat{\widetilde{q}}(s) (-1)^{y \cdot s_1\dots s_k} \sum_{x \in \{0,1\}^{n-k}} (-1)^{x \cdot s_{k+1,\dots,n}} \\
&=& 2^{n-k} \sum_{s,s_{k+1,\dots,n}=0^{n-k}} \widehat{\widetilde{q}}(s) (-1)^{y \cdot s_1\dots s_k}.
\eeas
Although in general the sum could contain up to $2^n$ terms, we only need to include those terms where $\widehat{\widetilde{q}}(s) \neq 0$. For each $y$, $S_y$ can therefore be computed exactly in $N \poly(n) = n^{O(\log(\alpha/\delta)/\epsilon)}$ time. It remains to show part 3 of the plan sketched in the introduction to this section: that knowledge of the marginals of $\widetilde{q}$ allows us to sample from a distribution approximating $p$.


\subsection{Sampling from an approximate probability distribution}

We now show that, in a quite general setting, if we can compute the marginals of an approximation $p'$ to a probability distribution $p$, we can approximately sample from $p$. Note that this task is apparently rather similar to one considered by Schwarz and Van den Nest~\cite{schwarz13}, who showed that certain quantum circuit families -- such as IQP circuits -- with sparse output distributions can be simulated classically, by using the Kushilevitz-Mansour algorithm~\cite{kushilevitz91} to approximately learn the corresponding Fourier coefficients, then showing that a probability distribution close to the corresponding approximate probability distribution can be sampled from exactly. However, in~\cite{schwarz13} it was sufficient to show that, given a probability distribution $p$ with at most $\poly(n)$ nonzero probabilities, each determined up to additive error $O(1/\poly(n))$, we can approximately sample from $p$. Here we have something harder to work with: that the distribution we have approximates $p$ up to constant overall $\ell_1$ error.

Fix integer $n \ge 0$. Imagine we have access to marginals of some $p' \in \R^{2^n}$ such that $\| p' - p\|_1 \le \delta$ (for $n=0$, $p'$ is just a real number), and that $\sum_x p'_x > 0$. Here ``access'' means that we can exactly compute sums of the form $S_y := \sum_{x,x_{1\dots k} = y} p'_x$ for each $k \in \{0,\dots,n\}$ and each $y \in \{0,1\}^k$, where for $k=0$ we consider the empty string $y$ and define $S = \sum_x p'_x$. We would like to sample from a probability distribution approximating $p$. Note that $p'$ may not be a probability distribution itself.

We use the following procedure, which is a ``truncated'' version of a standard procedure for sampling from a probability distribution, given access to its marginals:
\begin{enumerate}
\item Set $y$ to the empty string.
\item For $i=1,\dots,n$:
\begin{enumerate}
\item If $S_{yz} < 0$ for some $z \in \{0,1\}$, set $y \leftarrow y\bar{z}$, where $\bar{z} = 1-z$.
\item Otherwise: with probability $S_{y0} / S_y$, set $y \leftarrow y0$; otherwise, set $y \leftarrow y1$.
\end{enumerate}
\item Return $y$.
\end{enumerate}
We observe that, at each step of the procedure, there can be at most one $z \in \{0,1\}$ such that $S_{yz} < 0$. Otherwise, we would have $S_y < 0$, and hence $y$ would not have been picked at the previous step. Therefore this procedure defines a probability distribution Alg$(p')$ on $n$-bit strings, for any $p'$ such that $S > 0$. Crucially, we can show that Alg$(p') \approx p$:

\begin{lem}
\label{lem:approx}
Let $p$ be a probability distribution on $\{0,1\}^n$. Assume that $p':\{0,1\}^n \rightarrow \R$ satisfies $\| p' - p\|_1 \le \delta$ for some $\delta < 1$. Then $\|\operatorname{Alg}(p') - p \|_1 \le 4\delta/(1-\delta)$.
\end{lem}

We defer the proof of Lemma \ref{lem:approx} to Appendix \ref{app:sampling}.

All that remains to prove Theorem \ref{thm:main} is to put all the pieces together. The overall algorithm starts by approximating and storing enough Fourier coefficients of $\widetilde{q}$ required to ensure that $\|\operatorname{Alg}(\widetilde{q}) - p \|_1 \le \delta$. From Lemma \ref{lem:approx} and the discussion in previous sections, this can be achieved in time $n^{O(\log(\alpha/\delta)/\epsilon)}$. Then each sample from $\operatorname{Alg}(\widetilde{q})$ can be produced in time $\poly(n)$. This completes the proof.


\section{Extensions}

\subsection{Other algorithms}

There is not that much about the classical simulation approach proposed here which is specific to IQP circuits. Indeed, it will work for any class of circuits for which the output distribution is sufficiently anticoncentrated, and for which we can classically compute the Fourier coefficients of the output distribution.

{\bf Simon's algorithm.} Simon's quantum algorithm solves a certain oracular problem using exponentially fewer queries to the oracle than any possible classical algorithm~\cite{simon97}. In Simon's problem we are given access to a function $f:\{0,1\}^n \rightarrow Y$ for some set $Y$, and are promised  that there exists $t \in \{0,1\}^n$ such that $f(x) = f(y)$ if and only if $x + y = t$, where addition is bitwise modulo 2. Our task is to determine $t$. Simon's algorithm solves this problem using $O(n)$ evaluations of $f$, whereas any classical algorithm requires $\Omega(2^{n/2})$ evaluations. The output probability distribution of the algorithm is uniformly random over bit-strings $x \in \{0,1\}^n$ such that $x \cdot t = 0$. This distribution is sufficiently anticoncentrated for the above algorithm to work ($\alpha = 2$), and the Fourier coefficients of the output probability distribution $p$ can easily be calculated; $\hat{p}(0^n) = 2^{-n}$, and for $s \neq 0^n$,
\[ \hat{p}(s) = \frac{1}{2^{2n-1}} \sum_{x,x \cdot t=0} (-1)^{x \cdot s} = \frac{1}{2^{2n}} \sum_x (1+(-1)^{x \cdot t}) (-1)^{x \cdot s} = \frac{1}{2^{2n}} \sum_x (-1)^{x \cdot(s+t)} = \frac{\delta_{st}}{2^n}. \]
So we can evaluate $\hat{p}(s)$ by determining whether $s=t$, which can be done efficiently (for a given $s$).

{\bf Other algorithms?} Assume that we have the ability to exactly compute arbitrary probabilities $p_x$ in $\poly(n)$ time (note that this does not necessarily give us the ability to sample from $p$). For the above approach to work, we would like to approximate $2^n \hat{p}(s) = \sum_x p_x (-1)^{s \cdot x}$ up to additive accuracy $\delta$. In general, we will not be able to do this efficiently; for example, imagine $p_x = 1$ for some unique $x$, and all other probabilities are 0. Then $\hat{p}(s)$ only depends on one $x$, which we do not know in advance. A similar argument still holds for relatively anticoncentrated distributions. On the other hand, by a similar argument to that used to approximate $\hat{p}(s)$ for IQP circuits, we can achieve a suitable level of approximation whenever we are able to exactly compute the Fourier coefficients of the output state $\ket{\psi}$. Indeed, it is even sufficient to approximate $\bracket{s}{H^{\otimes n}}{\psi}$ up to very high accuracy.

One particular case which is tempting to address is the ``quantum approximate optimization algorithm'' (QAOA) invented by Farhi, Goldstone and Gutmann~\cite{farhi14,farhi14a}. This algorithm has been proposed to offer a route towards proving quantum supremacy~\cite{farhi16}. In the simplest version of the algorithm, the first step is to produce the state $\ket{\psi} = e^{-iB} e^{-iC} \ket{+}^{\otimes n}$, where $B = \beta \sum_i X_i$, $C = \gamma \sum_i C_i$ for some coefficients $\beta$, $\gamma$, where $X_i$ is Pauli-X on the $i$'th qubit, and each matrix $C_i$ is diagonal and only acts nontrivially on $O(1)$ qubits. The second step is to measure $\ket{\psi}$ in the computational basis to sample from a hopefully interesting distribution. The structure of the QAOA is very similar to an IQP circuit, and hardness of simulating the algorithm classically, up to small relative error, can be proven under similar assumptions to those for IQP circuits~\cite{farhi16}. We can think of $e^{-i \beta B} = \sm{\cos \beta & -i \sin \beta \\ -i \sin \beta & \cos \beta}$ as a kind of variant $H$ gate. In this case we can approximate $\bracket{s}{H^{\otimes n}}{\psi}$, but not to a sufficiently high level of accuracy for the above approach to work.


\subsection{Reducing the anticoncentration requirement}

One apparently non-ideal aspect of our results on simulating IQP circuits is the dependence on $\alpha$, meaning that we only obtain a polynomial-time classical simulation when the output probability distribution of the circuit is rather spread out. Interestingly, it was shown by Schwarz and Van den Nest~\cite{schwarz13} that IQP circuits can be simulated classically (with a similar notion of simulation to that considered here) if the output probability distribution $p$ is close to sparse. That is, if there exists a distribution $p'$ such that $\|p - p'\|_1 \le \delta$ for some small fixed $\delta$, and such that $p'$ only contains $t = \poly(n)$ nonzero probabilities (``$p$ is $\epsilon$-close to $t$-sparse''). This seems close to being a converse to the condition considered here, that $\sum_x p_x^2 \le \alpha 2^{-n}$ for $\alpha = O(1)$. If this were the case, we would have shown that noisy IQP circuits can always be simulated (if we can simulate a noiseless IQP circuit, we can simulate a noisy one, by sampling from the output distribution and then applying noise to the sample). From our results on fault-tolerance below, we would not expect this to be possible.

However, the constraint used here is not precisely the converse of that in~\cite{schwarz13}. Consider an IQP circuit $\mathcal{C}$ whose diagonal part consists of CZ gates on qubits $(1,2), (3,4),\dots,(k-1,k)$. Then $p$ is uniformly distributed over the set of bit-strings $x$ such that $x_i = 0$ for $i \in \{k+1,\dots,n\}$. So $p$ is $2^k$-sparse, but far from $L$-sparse for any $L \le 2^{k-1}$, for example. Further, $\sum_x p_x^2 = 2^{-k}$. If we take $k = n/2$, neither the present simulation method nor the method of~\cite{schwarz13} gives an efficient algorithm.



\section{Fault-tolerance}
\label{sec:ft}

We now show that the type of depolarising noise considered in this work can be dealt with using purely classical ideas from the theory of error-correcting codes. That is, we show that for any IQP circuit $\mathcal{C}$ with output probability distribution $p$, we can write down a corresponding IQP circuit $\mathcal{C}'$ such that, if depolarising noise is applied to every qubit of the output of $\mathcal{C}'$, we can still sample from a distribution which is close to $p$ up to arbitrarily small $\ell_1$ distance.

Let $M$ be an $n \times m$ matrix over $\F_2$, $m \ge n$, such that the rows of $M$ are linearly independent. For any function $f:\{0,1\}^n \rightarrow \C$,  define the function $f_M:\{0,1\}^m \rightarrow \C$ by $f_M(x) = f(Mx)$, where $Mx$ denotes matrix multiplication over $\F_2$. Then it is easy to see that, for any $s \in \{0,1\}^m$ such that $s = M^T t$ for some $t \in \{0,1\}^n$, $\widehat{f_M}(s) = \hat{f}(t)$; and otherwise, $\widehat{f_M}(s) = 0$. Indeed, if we define a function $g:\{0,1\}^m \rightarrow \C$ by
\[ \hat{g}(s) = \begin{cases} \hat{f}(t) & \text{if $s = M^T t$ for some $t \in \{0,1\}^n$}\\ 0 & \text{otherwise} \end{cases}, \]
then
\[ g(x) = \sum_s \hat{g}(s) (-1)^{s \cdot x} = \sum_t \hat{f}(t) (-1)^{(M^T t) \cdot x} = \sum_t \hat{f}(t) (-1)^{t \cdot (Mx)} = f(Mx), \]
where we use linear independence of the rows of $M$ in the second equality. So $g = f_M$, and equivalently $\widehat{f_M} = \hat{g}$.

This implies that a linear transformation $s \mapsto M^T s$ of the output probability distribution of a unitary operation $\mathcal{C} = H^{\otimes n} D H^{\otimes n}$ can be achieved by applying a corresponding linear transformation to the diagonal part of the circuit. If $\mathcal{C}$ is an IQP circuit where $D$ is made up of $\poly(n)$ diagonal gates, this transformation can be performed efficiently, i.e.\ in time $\poly(n)$.
%
%
Indeed, the diagonal part of any IQP circuit can be written as
\be \label{eq:xprogram} D = e^{i \sum_{j=1}^\ell \theta_j \prod_{k=1}^n Z_j^{C_{jk}}} \ee
for some real coefficients $\theta_j$, where $C$ is an $\ell \times n$ matrix over $\F_2$, and $Z_j$ denotes a Pauli-Z operation acting on the $j$'th qubit. This formalism was introduced in~\cite{shepherd09} and is known as an ``X-program''\footnote{The ``X'' comes from the conjugation by Hadamards replacing the Pauli-Z gates with Pauli-X gates.}. If we replace $C$ with $CM$, we obtain a new circuit whose diagonal part is
\[ D_M = e^{i \sum_{j=1}^\ell \theta_j \prod_{k=1}^m Z_j^{(CM)_{jk}}}. \]
Then, using the fact that
\[ \bracket{x}{\prod_{k=1}^m Z_j^{(CM)_{jk}}}{x} =  (-1)^{\sum_{k=1}^m (CM)_{jk} x_k} =  (-1)^{\sum_{k=1}^n C_{jk} (Mx)_k} = \bracket{Mx}{\prod_{k=1}^n Z_j^{C_{jk}}}{Mx}, \]
we have $\bracket{x}{D_M}{x} = \bracket{Mx}{D}{Mx}$. Very similar, and more general, ideas about transformation of IQP circuits were introduced in~\cite{shepherd09,shepherd10}, albeit in slightly different language.

Applying linear transformations over $\F_2$ to the output distribution of $\mathcal{C}$ allows us to use techniques from the classical theory of error-correcting codes to combat noise. Let $M$ be the generator matrix of an error-correcting code of length $m$ and dimension $n$. The noise operation $\mathcal{N}_\epsilon$ we consider corresponds to sampling a bit-string $t = M^Ts$, then flipping each bit of $t$ with independent probability $\epsilon/2$ to produce a new bit-string $\widetilde{t}$; we write $\widetilde{t} \sim_\epsilon t$ for the distribution on noisy bit-strings $\widetilde{t}$. So if $M$ has an efficient decoding algorithm, given a noisy sample $\widetilde{t} \in \{0,1\}^m$, we can apply the decoding algorithm to produce a new sample $s'$. If $M$ has good error-correction properties, $s' = s$ with high probability.

More formally, we consider the distribution $p'$ on bit-strings $s' \in \{0,1\}^n$ produced by the following procedure:
\begin{enumerate}
\item Produce a sample $s \in \{0,1\}^n$ from a distribution $p$ where $p_s = |\hat{f}(s)|^2$.
\item Set $t = M^T s$.
\item Flip each bit of $t$ with independent probability $\epsilon/2$ to produce $\widetilde{t}$.
\item Apply the decoding algorithm $D$ for $M$ to produce $s' = D(\widetilde{t})$.
\end{enumerate}
Then we would like to show that $\| p - p' \|_1 \le \delta$ for arbitrarily small $\delta$. If it holds that, for all encoded bit-strings $x$, $\Pr_{y \sim_\epsilon x}[D(y)\neq D(x)]\le\delta/2$, then $\|(p')^{(s)} - p^{(s)} \|_1 \le \delta$ for all original bit-strings $s$, where $p^{(s)}$ is the point distribution on bit-strings $s \in \{0,1\}^n$ such that $p^{(s)}(s) = 1$, and $(p')^{(s)}$ is the corresponding distribution on decoded noisy bit-strings. Taking the average over bit-strings $s$ according to $p$, we have $\| p - p' \|_1 \le \delta$ by convexity.

There are classical error-correcting codes which achieve $\Pr_{y \sim_\epsilon x}[D(y)\neq D(x)] \rightarrow 0$ for any $\epsilon < 1$ with only a modest overhead. Indeed, Shannon's noisy channel coding theorem for the binary symmetric channel states that an exponentially small (in $n$) failure probability can be achieved (nonconstructively) for any $\epsilon < 1$ by taking $m = c_\epsilon n$ for some $c_\epsilon$ that depends only on $\epsilon$. Explicit codes which almost achieve Shannon's nonconstructive bound, and have efficient decoding algorithms, are known; for example, low-density parity check codes and certain concatenated codes~\cite{richardson08}. So the overhead need only be a constant factor.

For our purposes, it would even be sufficient to use a simple repetition code, where each bit is encoded in $r$ bits. Here $M$ is just $r$ copies of the identity matrix (for $r$ odd) and the decoding algorithm takes a majority vote. The probability that a bit is decoded incorrectly is the same as the probability that more than $r/2$ bits are flipped, which is
\[ \sum_{i > r/2} \binom{r}{i} (\epsilon/2)^i (1-\epsilon/2)^{r-i} \le 2^r (\epsilon/2)^{r/2} (1-\epsilon/2)^{r/2} = \left( \epsilon(2-\epsilon) \right)^{r/2}, \]
so for any $\epsilon < 1$ the probability that an individual bit is decoded incorrectly is exponentially small in $r$. So, using a repetition code, it is sufficient to take $r = O(\log n)$ for all $n$ bits to be decoded successfully except with low probability.

The IQP circuits produced using error-correcting codes could be more complex than the original circuits, as they may involve gates acting on up to $n$ qubits. However, in some cases these gates can then be replaced with gates acting on only $O(1)$ qubits each. For example, consider the family of circuits where $\theta_j$ is restricted to be a multiple of $\pi/8$ in (\ref{eq:xprogram}), which is one of the cases shown hard to simulate in~\cite{bremner16} (assuming some complexity-theoretic conjectures). It is shown in~\cite{shepherd10} that any gate in such a circuit, even acting on all $n$ qubits, can be replaced with $\poly(n)$ gates from the same family acting on at most 3 qubits each without changing the output of the circuit.



\subsection*{Acknowledgements}

AM was supported by an EPSRC Early Career Fellowship (EP/L021005/1). MJB has received financial support from the Australian Research Council via the Future Fellowship scheme (grant FT110101044) and acknowledges support as a member of the ARC Centre of Excellence for Quantum Computation and Communication Technology (CQC2T), project number CE170100012. We would like to thank Richard Jozsa, Sergio Boixo, Eleanor Rieffel, Ryan Mann and Juan Bermejo-Vega for helpful comments. No new data were created during this study.


\appendix

\section{Anticoncentration bound}
\label{app:anticonc}

\begin{replem}{lem:anticonc}
Let $\mathcal{C}$ be a random sparse IQP circuit. Then $\E_{\mathcal{C}}[|\bracket{0}{\mathcal{C}}{0}|^2] = 2^{-n}$ and, for a large enough constant $\gamma$, $\E_{\mathcal{C}}[|\bracket{0}{\mathcal{C}}{0}|^4] \le 5 \cdot 2^{-2n}$.
\end{replem}

\begin{proof}
It is easy to see from symmetry arguments~\cite{bremner16} that $\E_{\mathcal{C}}[|\bracket{0}{\mathcal{C}}{0}|^2] = 2^{-n}$. So all that remains is to get a bound on $\E_{\mathcal{C}}[|\bracket{0}{\mathcal{C}}{0}|^4]$.

Let $\alpha_{ij} \in \{0,\dots,3\}$ be the number of times the gate $(1,1,1,\omega)$ is applied across qubits $i$ and $j$. It is shown in the last appendix of~\cite{bremner16} that, for any distribution on the $\alpha_{ij}$ coefficients,
\[ |\bracket{0}{\mathcal{C}}{0}|^4 = 2^{-4n} \sum_{w,x,y \in \{0,1\}^n} \prod_{i < j} \E_{\alpha_{ij}} \left[\omega^{\alpha_{ij}(w_i(y_j-x_j)+x_i(y_j-w_j)+y_i(w_j+x_j)-2y_iy_j)} \right]. \]
For any coefficients $\beta_{ij}$, we have
\[ \E_{\alpha_{ij}}[\omega^{\alpha_{ij} \beta_{ij}}] = (1-p)1 + p\,\E_{\alpha_{ij} \sim \mathcal{U}}[\omega^{\alpha_{ij}\beta_{ij}}] = \begin{cases} 1 & \text{if $\beta_{ij} \equiv 0$ mod 4}\\ 1-p & \text{otherwise,} \end{cases} \]
where $\mathcal{U}$ is the uniform distribution on $\{0,\dots,3\}$, recalling that $p$ is the probability that a gate is applied across qubits $i$ and $j$. As $w,x,y \in \{0,1\}^n$, the expression $F_{ij}(w,x,y) := w_i(y_j-x_j)+x_i(y_j-w_j)+y_i(w_j+x_j)-2y_iy_j$ is zero mod 4 if and only if it equals zero. We therefore have
\[ |\bracket{0}{\mathcal{C}}{0}|^4 = 2^{-4n} \sum_{w,x,y\in\{0,1\}^n} (1-p)^{|\{(i<j): F_{ij}(w,x,y) \neq 0 \}|}. \]
It can be checked that, for any $k \in \{1,\dots,n\}$, $F_{ij}(w,x,y) \neq 0$ if and only if $F_{ij}(w^k,x^k,y^k) \neq 0$, where $w^k$ is the bit-string produced from $w$ by flipping the $k$'th bit. By flipping bits of $w$, we can therefore assume that $w=0^n$ and obtain
\[ |\bracket{0}{\mathcal{C}}{0}|^4 = 2^{-3n} \sum_{x,y\in\{0,1\}^n} (1-p)^{|\{(i<j): F_{ij}(0^n,x,y) \neq 0 \}|}. \]
For a given pair $(i<j)$, $F_{ij}(0^n,x,y) = x_i y_j + y_i x_j - 2 y_i y_j \neq 0$ if and only if the pairs of strings $(x_i x_j, y_i y_j)$ are in the following set:
\[ \{(00,11),(01,10),(01,11),(10,01),(10,11),(11,01),(11,10) \}. \]
Put another way, the strings $(x_i y_i, x_j y_j)$ should be in the following set:
\[ \{ (01,01), (01,10), (01,11), (10,01), (11,01), (10,11), (11,10) \} \]
Define integers $a$, $b$, $c$ by
\[ a = |\{i:x_i=y_i=0\}|,\;\; b = |\{i:x_i=0,y_i=1\}|,\;\; c = |\{i:x_i=1,y_i=0\}|,\;\; d = |\{i:x_i=y_i=1\}|. \]
Then
\[ |\{(i<j): F_{ij}(0^n,x,y) \neq 0 \}| = \binom{b}{2} + bc + bd + cd. \]
So
\[ |\bracket{0}{\mathcal{C}}{0}|^4 \le 2^{-3n} \sum_{b,c,d=0}^n N_{bcd} (1-p)^{bc+bd+cd}, \]
where $N_{bcd} = \binom{n}{b} \binom{n-b}{c} \binom{n-b-c}{d}$ is the number of pairs $(x,y)$ with the correct numbers of combinations of bits (so $|\{i:x_i=0,y_i=1\}| = b$, etc.), and we have simplified by removing the $\binom{b}{2}$ term, which can only make the inequality looser. 
We now split into cases. Let $\alpha$ be a constant such that $\binom{n}{\alpha n} \le 2^{n/3} / (n+1)$ (for example, $\alpha = 1/20$ works for large enough $n$). Then, as $N_{bcd} \le \binom{n}{b} \binom{n}{c} \binom{n}{d}$, all terms in the sum such that $\max\{b,c,d\} \le \alpha n$ are bounded by $2^n / (n+1)^3$. Now consider a term in the sum such that at least one of $b,c,d$ is larger than $\alpha n$ (assume $b$ wlog). Then
\[ N_{bcd} = \binom{n}{b} \binom{n-b}{c} \binom{n-b-c}{d} \le 2^n n^{c+d}, \]
so
\[ N_{bcd} (1-p)^{bc+bd+cd} \le 2^n n^{c+d} (1-p)^{b(c+d)} \le 2^n (n(1-p)^{\alpha n})^{c+d} \le 2^n (n e^{-\alpha \gamma \ln n})^{c+d}. \]
Taking $\gamma = 4/\alpha$, this is upper-bounded by $2^n n^{-3}$ whenever $c \ge 1$ or $d \ge 1$. It remains to consider the cases where $c=0,d=0$: the sum resulting from these is bounded by $2^{-3n} \sum_{b=0}^n \binom{n}{b} = 2^{-2n}$. Multiplying by 3, to allow for the choice of each of $b,c,d$ as the large value, the entire sum, and hence $|\bracket{0}{\mathcal{C}}{0}|^4$, is bounded by $5 \cdot 2^{-2n}$. This completes the proof.
\end{proof}


\section{Sampling from an approximate distribution}
\label{app:sampling}

In this appendix we prove Lemma \ref{lem:approx} regarding the behaviour of the algorithm Alg applied to approximate probability distributions $p'$. To do so, we define a closely related, but somewhat easier to analyse, procedure Fix$(p)$ as follows for vectors $p \in \R^{2^n}$, integer $n \ge 0$, such that $\sum_x p_x > 0$. First, for $n=0$, Fix$(p)=p$. For $n \ge 1$, writing $p = \sm{a\\b}$ for some $a,b \in \R^{2^{n-1}}$,
\[
\operatorname{Fix}(p) =
\begin{cases}
\begin{pmatrix} \operatorname{Fix}(a) \\ \operatorname{Fix}(b) \end{pmatrix} & \text{ if $\sum_x a_x > 0$ and $\sum_x b_x > 0$}\\
\frac{\sum_x p_x}{\sum_x a_x} \begin{pmatrix} \operatorname{Fix}(a) \\ 0 \end{pmatrix} & \text{ if $\sum_x a_x > 0$ and $\sum_x b_x \le 0$}\\
\frac{\sum_x p_x}{\sum_x b_x} \begin{pmatrix} 0 \\ \operatorname{Fix}(b) \end{pmatrix} & \text{ if $\sum_x a_x \le 0$ and $\sum_x b_x > 0$.}
\end{cases}
\]
Note that we cannot have $\sum_x a_x \le 0$ and $\sum_x b_x \le 0$ simultaneously because $\sum_x p_x > 0$, and further that in the recursive definition, Fix is never applied to an ``illegal'' vector whose entry sum is nonpositive. Up to scaling, Fix is equivalent to Alg:

\begin{lem}
\label{lem:fixalg}
For any integer $n \ge 0$, and any $p \in \R^{2^n}$ such that $\sum_x p_x > 0$, $\operatorname{Fix}(p) = (\sum_x p_x) \operatorname{Alg}(p)$.
\end{lem}

\begin{proof}
The proof is by induction on $n$. The claim clearly holds for $n=0$, where Fix$(p) = p = p \cdot 1 = p \operatorname{Alg}(p)$. For $n \ge 1$, writing $p = \sm{a\\b}$ for some $a,b \in \R^{2^{n-1}}$, by the definition of Alg we have

\[
\operatorname{Alg}(p) =
\begin{cases}
\frac{1}{\sum_x p_x} \begin{pmatrix} (\sum_x a_x) \operatorname{Alg}(a) \\ (\sum_x b_x)  \operatorname{Alg}(b) \end{pmatrix} & \text{ if $\sum_x a_x > 0$ and $\sum_x b_x > 0$}\\
\frac{1}{\sum_x a_x} \begin{pmatrix} \operatorname{Alg}(a) \\ 0 \end{pmatrix} & \text{ if $\sum_x a_x > 0$ and $\sum_x b_x \le 0$}\\
\frac{1}{\sum_x b_x} \begin{pmatrix} 0 \\ \operatorname{Alg}(b) \end{pmatrix} & \text{ if $\sum_x a_x \le 0$ and $\sum_x b_x > 0$.}
\end{cases}
\]
So, using the inductive hypothesis,
\[
(\sum_x p_x) \operatorname{Alg}(p) =
\begin{cases}
\begin{pmatrix} \operatorname{Fix}(a) \\ \operatorname{Fix}(b) \end{pmatrix} & \text{ if $\sum_x a_x > 0$ and $\sum_x b_x > 0$}\\
\frac{\sum_x p_x}{\sum_x a_x} \begin{pmatrix} \operatorname{Fix}(a) \\ 0 \end{pmatrix} & \text{ if $\sum_x a_x > 0$ and $\sum_x b_x \le 0$}\\
\frac{\sum_x p_x}{\sum_x b_x} \begin{pmatrix} 0 \\ \operatorname{Fix}(b) \end{pmatrix} & \text{ if $\sum_x a_x \le 0$ and $\sum_x b_x > 0$,}
\end{cases}
\]
which is equal to $\operatorname{Fix}(p)$ as required.
\end{proof}

Therefore, if we apply Alg to some approximate distribution $p'$, the final probability distribution sampled from is precisely Fix$(p') / S$, where $S = \sum_x p'_x$. This shows, in particular, that for all $p$ such that $\sum_x p_x > 0$, Fix$(p)_x \ge 0$ for all $x$.

\begin{lem}
\label{lem:fix}
For any integer $n \ge 0$, and any $p \in \R^{2^n}$ such that $\sum_x p_x > 0$, $\|\operatorname{Fix}(p) - p \|_1 = 2\sum_{x,p_x < 0} |p_x|$.
\end{lem}

\begin{proof}
We first show that the following claims imply the lemma: for all $p \in \R^{2^n}$ such that $\sum_x p_x > 0$, then
\begin{enumerate}
\item For all $x$ such that $p_x \ge 0$, $0 \le \operatorname{Fix}(p)_x \le p_x$;
\item For all $x$ such that $p_x < 0$, $\operatorname{Fix}(p)_x = 0$;
\item $\sum_x \operatorname{Fix}(p)_x = \sum_x p_x$.
\end{enumerate}
Indeed, assuming these claims, we have
\begin{eqnarray*}
\|\operatorname{Fix}(p) - p \|_1 &=& \sum_x | p_x - \operatorname{Fix}(p)_x | = \sum_{x, p_x \ge 0} (p_x - \operatorname{Fix}(p)_x) - \sum_{x, p_x < 0} p_x \\
&=& \sum_{x, p_x \ge 0} p_x - \sum_x p_x - \sum_{x,p_x < 0} p_x = 2\sum_{x,p_x < 0} |p_x|
\end{eqnarray*}
as desired, where the second equality uses claims 1 and 2, and the third uses claims 2 and 3. It remains to prove the claims. Claim 1 follows from observing that Fix never changes the sign of an element of $p$, and at each step modifies elements by either zeroing them, or rescaling them by a scaling factor upper-bounded by 1. Claim 2 follows from considering the last-but-one step of Fix, where it is applied to vectors of the form $\sm{ \alpha \\ \beta}$ with $\alpha + \beta > 0$; if either of $\alpha$ or $\beta$ is negative, it will be zeroed by Fix. Claim 3 is shown by induction: it clearly holds for $n=0$ as $p \ge 0$ and Fix$(p)=p$, and for $n \ge 1$, assuming the inductive hypothesis for vectors on $\R^{2^{n-1}}$ and inspecting the definition of Fix shows that in all three cases $\sum_x \operatorname{Fix}(p)_x = \sum_x p_x$. This completes the proof.
\end{proof}

We are finally ready to prove Lemma \ref{lem:approx}.

\begin{replem}{lem:approx}
Let $p$ be a probability distribution on $\{0,1\}^n$. Assume that $p':\{0,1\}^n \rightarrow \R$ satisfies $\| p' - p\|_1 \le \delta$ for some $\delta < 1$. Then $\|\operatorname{Alg}(p') - p \|_1 \le 4\delta/(1-\delta)$.
\end{replem}

\begin{proof}
First, we have
\[ |1 - \sum_x p'_x| = | \sum_x p_x - p'_x| \le \sum_x |p_x- p'_x| \le \delta, \]
so $S := \sum_x p'_x \ge 1-\delta > 0$ and hence $p'$ satisfies the preconditions of Lemmas \ref{lem:fixalg} and \ref{lem:fix}. So
\beas \|\operatorname{Alg}(p') - p \|_1 &=& \| \frac{1}{S} \operatorname{Fix}(p') - p\|_1\\
&\le& \frac{1}{S} \| \operatorname{Fix}(p') - p'\|_1 + \frac{1}{S} \| p' - p \|_1 + \|\frac{1}{S} p - p\|_1\\
&\le& \frac{2 \sum_{x,p'_x < 0} |p'_x|}{S} + \frac{\delta}{S} + \frac{1}{S} - 1\\
&\le& \frac{2\delta}{1-\delta}  + \frac{\delta+1}{1-\delta} - 1\\
&=& \frac{4\delta}{1-\delta},
\eeas
where the first equality is Lemma \ref{lem:fixalg}, the first inequality is the triangle inequality, the second is Lemma \ref{lem:fix}, and the third uses $\sum_{x,p'_x<0} |p'_x| \le \sum_{x,p'_x<0} |p_x - p'_x| \le \delta$.
\end{proof}


\bibliographystyle{myplain2-doi}
\bibliography{poweriqp}

\begin{thebibliography}{10}

\bibitem{aaronson13}
S.~Aaronson and A.~Arkhipov.
\newblock \href{https://doi.org/10.1145/1993636.1993682}{The computational
  complexity of linear optics}.
\newblock {\em Theory of Computing}, 9\penalty0 (4):\penalty0 143--252, 2013.
\newblock {\tt arXiv:1011.3245}.

\bibitem{alon94}
N.~Alon, A.~Frieze, and D.~Welsh.
\newblock \href{https://doi.org/10.1109/SFCS.1994.365708}{Polynomial time
  randomized approximation schemes for the {T}utte polynomial of dense graphs}.
\newblock In {\em Proc. 35\textsuperscript{th} Annual Symp. Foundations of
  Computer Science}, 1994, page~24.

\bibitem{arkhipov15}
A.~Arkhipov.
\newblock \href{https://doi.org/10.1103/PhysRevA.92.062326}{{BosonSampling} is
  robust to small errors in the network matrix}.
\newblock {\em Phys. Rev. A}, 92:\penalty0 062326, 2015.
\newblock {\tt arXiv:1412.2516}.

\bibitem{arora99}
S.~Arora, D.~Karger, and M.~Karpinski.
\newblock \href{https://doi.org/10.1145/225058.225140}{Polynomial time
  approximation schemes for dense instances of {NP}-hard problems}.
\newblock {\em Journal of Computer and System Sciences}, 58:\penalty0 193--210,
  1999.

\bibitem{beals13}
R.~Beals, S.~Brierley, O.~Gray, A.~Harrow, S.~Kutin, N.~Linden, D.~Shepherd,
  and M.~Stather.
\newblock \href{https://doi.org/10.1098/rspa.2012.0686}{Efficient distributed
  quantum computing}.
\newblock {\em Proc. Roy. Soc. A}, 469:\penalty0 20120686, 2013.
\newblock {\tt arXiv:1207.2307}.

\bibitem{bermejovega17}
J.~Bermejo-Vega, D.~Hangleiter, M.~Schwarz, R.~Raussendorf, and J.~Eisert.
\newblock Architectures for quantum simulation showing quantum supremacy, 2017.
\newblock {\tt arXiv:1703.00466}.

\bibitem{boixo16}
S.~Boixo, S.~Isakov, V.~Smelyanskiy, R.~Babbush, N.~Ding, Z.~Jian, J.~Martinis,
  and H.~Neven.
\newblock Characterizing quantum supremacy in near-term devices, 2016.
\newblock {\tt arXiv:1608.00263}.

\bibitem{bollobas80}
B.~Bollob{\'a}s.
\newblock \href{https://doi.org/10.1016/0012-365X(80)90054-0}{The distribution
  of the maximum degree of a random graph}.
\newblock {\em Discrete Mathematics}, 32:\penalty0 201--203, 1980.

\bibitem{brand16}
C.~Brand, H.~Dell, and M.~Roth.
\newblock Fine-grained dichotomies for the {T}utte plane and {B}oolean \#{CSP},
  2016.
\newblock {\tt arXiv:1606.06581}.

\bibitem{bremner11}
M.~Bremner, R.~Jozsa, and D.~Shepherd.
\newblock \href{https://doi.org/10.1098/rspa.2010.0301}{Classical simulation of
  commuting quantum computations implies collapse of the polynomial hierarchy}.
\newblock {\em Proc. Roy. Soc. Ser. A}, 467\penalty0 (2126):\penalty0 459--472,
  2011.
\newblock {\tt arXiv:1005.1407}.

\bibitem{bremner16}
M.~Bremner, A.~Montanaro, and D.~Shepherd.
\newblock \href{https://doi.org/10.1103/PhysRevLett.117.080501}{Average-case
  complexity versus approximate simulation of commuting quantum computations}.
\newblock {\em Phys. Rev. Lett.}, 117:\penalty0 080501, 2016.
\newblock {\tt arXiv:1504.07999}.

\bibitem{brown12}
W.~Brown and O.~Fawzi.
\newblock Scrambling speed of random quantum circuits, 2012.
\newblock {\tt arXiv:1210.6644}.

\bibitem{brown15}
W.~Brown and O.~Fawzi.
\newblock \href{https://doi.org/10.1007/s00220-015-2470-1}{Decoupling with
  random quantum circuits}.
\newblock {\em Comm. Math. Phys.}, 340\penalty0 (3):\penalty0 867--900, 2015.
\newblock {\tt arXiv:1307.0632}.

\bibitem{buhrman06a}
H.~Buhrman, R.~Cleve, M.~Laurent, N.~Linden, A.~Schrijver, and F.~Unger.
\newblock \href{https://doi.org/10.1109/FOCS.2006.50}{New limits on
  fault-tolerant quantum computation}.
\newblock In {\em Proc. 47\textsuperscript{th} Annual Symp. Foundations of
  Computer Science}, 2006, pages 411--419.
\newblock {\tt quant-ph/0604141}.

\bibitem{curticapean15}
R.~Curticapean.
\newblock \href{https://doi.org/10.1007/978-3-662-47672-7_31}{Block
  interpolation: A framework for tight exponential-time counting complexity}.
\newblock In {\em Proc. 42\textsuperscript{nd} {I}nternational {C}onference on
  {A}utomata, {L}anguages and {P}rogramming (ICALP'15)}, 2015, pages 380--392.
\newblock {\tt arXiv:1511.02910}.

\bibitem{diestel10}
R.~Diestel.
\newblock {\em Graph Theory}.
\newblock Springer, 2010.

\bibitem{dubhashi09}
D.~Dubhashi and A.~Panconesi.
\newblock {\em Concentration of measure for the analysis of randomized
  algorithms}.
\newblock Cambridge University Press, 2009.

\bibitem{farhi14}
E.~Farhi, J.~Goldstone, and S.~Gutmann.
\newblock A quantum approximate optimization algorithm, 2014.
\newblock {\tt arXiv:1411.4028}.

\bibitem{farhi14a}
E.~Farhi, J.~Goldstone, and S.~Gutmann.
\newblock A quantum approximate optimization algorithm applied to a bounded
  occurrence constraint problem, 2014.
\newblock {\tt arXiv:1412.6062}.

\bibitem{farhi16}
E.~Farhi and A.~Harrow.
\newblock Quantum supremacy through the {Quantum Approximate Optimization
  Algorithm}, 2016.
\newblock {\tt arXiv:1602.07674}.

\bibitem{fujii16}
K.~Fujii and S.~Tamate.
\newblock \href{https://doi.org/10.1038/srep25598}{Computational
  quantum-classical boundary of noisy commuting quantum circuits}.
\newblock {\em Scientific Reports}, 6:\penalty0 25598, 2016.
\newblock {\tt arXiv:1406.6932}.

\bibitem{gao17}
X.~Gao, S.-T. Wang, and L.-M. Duan.
\newblock \href{https://doi.org/10.1103/PhysRevLett.118.040502}{Quantum
  supremacy for simulating a translation-invariant {I}sing spin model}.
\newblock {\em Phys. Rev. Lett.}, 118:\penalty0 040502, 2017.
\newblock {\tt arXiv:1607.04947}.

\bibitem{hangleiter17}
D.~Hangleiter, M.~Kliesch, M.~Schwarz, and J.~Eisert.
\newblock \href{https://doi.org/10.1088/2058-9565/2/1/015004}{Direct
  certification of a class of quantum simulations}.
\newblock {\em Quantum Science and Technology}, 1\penalty0 (2), 2017.
\newblock {\tt arXiv:1602.00703}.

\bibitem{kalai16}
G.~Kalai.
\newblock \href{https://doi.org/10.1090/noti1380}{The quantum computer puzzle}.
\newblock {\em Notices of the AMS}, 63\penalty0 (5):\penalty0 508--516, 2016.
\newblock {\tt arXiv:1605.00992}.

\bibitem{kalai14}
G.~Kalai and G.~Kindler.
\newblock Gaussian noise sensitivity and {BosonSampling}, 2014.
\newblock {\tt arXiv:1409.3093}.

\bibitem{kempe08}
J.~Kempe, O.~Regev, F.~Unger, and R.~de~Wolf.
\newblock \href{https://doi.org/10.1007/978-3-540-70575-8_69}{Upper bounds on
  the noise threshold for fault-tolerant quantum computing}.
\newblock In {\em Proc. 35\textsuperscript{th} {I}nternational {C}onference on
  {A}utomata, {L}anguages and {P}rogramming (ICALP'08)}, 2008, pages 845--856.
\newblock {\tt arXiv:0802.1464}.

\bibitem{kushilevitz91}
E.~Kushilevitz and Y.~Mansour.
\newblock \href{https://doi.org/10.1137/0222080}{Learning decision trees using
  the {F}ourier spectrum}.
\newblock In {\em Proc. 23\textsuperscript{rd} Annual ACM Symp. Theory of
  Computing}, 1991, pages 455--464.

\bibitem{leverrier15}
A.~Leverrier and R.~Garc\'ia-Patr\'on.
\newblock Analysis of circuit imperfections in {BosonSampling}.
\newblock {\em Quantum Inf. Comput.}, 15:\penalty0 0489--0512, 2015.
\newblock {\tt arXiv:1309.4687}.

\bibitem{markov08}
I.~Markov and Y.~Shi.
\newblock \href{https://doi.org/10.1137/050644756}{Simulating quantum
  computation by contracting tensor networks}.
\newblock {\em SIAM J. Comput.}, 38:\penalty0 963--981, 2008.
\newblock {\tt quant-ph/0511069}.

\bibitem{marvian16}
M.~Marvian and D.~Lidar.
\newblock Error suppression for {H}amiltonian-based quantum computation using
  subsystem codes, 2016.
\newblock {\tt arXiv:1606.03795}.

\bibitem{morimae14}
T.~Morimae, K.~Fujii, and J.~Fitzsimons.
\newblock \href{https://doi.org/10.1103/PhysRevLett.112.130502}{On the hardness
  of classically simulating the one-clean-qubit model}.
\newblock {\em Phys. Rev. Lett.}, 112:\penalty0 130502, 2014.
\newblock {\tt arXiv:1312.2496}.

\bibitem{odonnell14}
R.~O'Donnell.
\newblock {\em Analysis of {B}oolean {F}unctions}.
\newblock Cambridge University Press, 2014.

\bibitem{preskill12}
J.~Preskill.
\newblock Quantum computing and the entanglement frontier, 2012.
\newblock {\tt arXiv:1203.5813}.

\bibitem{rahimikeshari16}
S.~Rahimi-Keshari, T.~Ralph, and C.~Caves.
\newblock \href{https://doi.org/10.1103/PhysRevX.6.021039}{Sufficient
  conditions for efficient classical simulation of quantum optics}.
\newblock {\em Phys. Rev. X}, 6:\penalty0 021039, 2016.
\newblock {\tt arXiv:1511.06526}.

\bibitem{razborov04}
A.~Razborov.
\newblock An upper bound on the threshold quantum decoherence rate.
\newblock {\em Quantum Inf. Comput.}, 4\penalty0 (3):\penalty0 222--228, 2004.
\newblock {\tt quant-ph/0310136}.

\bibitem{richardson08}
T.~Richardson and R.~Urbanke.
\newblock {\em Modern Coding Theory}.
\newblock Cambridge University Press, 2008.

\bibitem{schnorr86}
C.~Schnorr and A.~Shamir.
\newblock \href{https://doi.org/10.1145/12130.12156}{An optimal sorting
  algorithm for mesh connected computers}.
\newblock In {\em Proc. 18\textsuperscript{th} Annual ACM Symp. Theory of
  Computing}, 1986, pages 255--263.

\bibitem{schwarz13}
M.~Schwarz and M.~Van den Nest.
\newblock Simulating quantum circuits with sparse output distributions, 2013.
\newblock {\tt arXiv:1310.6749}.

\bibitem{shchesnovich15}
V.~Shchesnovich.
\newblock \href{https://doi.org/10.1103/PhysRevA.91.063842}{Tight bound on
  trace distance between a realistic device with partially indistinguishable
  bosons and the ideal {BosonSampling}}.
\newblock {\em Phys. Rev. A}, 91:\penalty0 063842, 2015.
\newblock {\tt arXiv:1501.00850}.

\bibitem{shepherd10}
D.~Shepherd.
\newblock Binary matroids and quantum probability distributions, 2010.
\newblock {\tt arXiv:1005.1744}.

\bibitem{shepherd09}
D.~Shepherd and M.~J. Bremner.
\newblock \href{https://doi.org/10.1098/rspa.2008.0443}{Temporally unstructured
  quantum computation}.
\newblock {\em Proc. Roy. Soc. Ser. A}, 465\penalty0 (2105):\penalty0
  1413--1439, 2009.
\newblock {\tt arXiv:0809.0847}.

\bibitem{simon97}
D.~R. Simon.
\newblock \href{https://doi.org/10.1137/S0097539796298637}{On the power of
  quantum computation}.
\newblock {\em SIAM J. Comput.}, 26:\penalty0 1474--1483, 1997.

\bibitem{toda91}
S.~Toda.
\newblock \href{https://doi.org/10.1137/0220053}{{PP} is as hard as the
  polynomial-time hierarchy}.
\newblock {\em SIAM J. Comput.}, 20\penalty0 (5):\penalty0 865--877, 1991.

\bibitem{virmani05}
S.~Virmani, S.~Huelga, and M.~Plenio.
\newblock \href{https://doi.org/10.1103/PhysRevA.71.042328}{Classical
  simulability, entanglement breaking, and quantum computation thresholds}.
\newblock {\em Phys. Rev. A}, 71:\penalty0 042328, 2005.
\newblock {\tt quant-ph/0408076}.

\end{thebibliography}

\end{document}